%% file: arxiv_submission_March.tex
\documentclass[onecolumn]{IEEEtran}

\usepackage{amsmath,amsfonts}
\usepackage{mathtools}
\usepackage{graphicx,xcolor,cite} 
\usepackage{bm}
\usepackage{amsthm,hyperref}
\usepackage{caption}
\newtheorem{theorem}{Theorem}

\definecolor{darkgreen}{rgb}{0.0, 0.4, 0.0}

\DeclareMathAlphabet\mathbfcal{OMS}{cmsy}{b}{n}

\DeclareMathOperator*{\st}{subject~to}
\DeclareMathOperator*{\Tr}{Tr}
\DeclareMathOperator*{\argmin}{arg\,min}

\newcommand{\norm}[1]{\left\lVert#1\right\rVert}
\newcommand{\normtwo}[1]{\left\lVert#1\right\rVert_{2}}
\newcommand{\frobenius}[1]{\left\lVert#1\right\rVert_{F}}
\newcommand{\hatbf}[1]{\widehat{\mathbf{#1}}}

\newcommand{\starbf}[1]{\mathbf{#1}^\star}
\newcommand{\starhatbf}[1]{\widehat{\mathbf{#1}}^\star}
\newcommand{\starhatbm}[1]{\widehat{\bm{#1}}^\star}

\newcommand{\hatbm}[1]{\widehat{\bm{#1}}}
\newcommand{\starbm}[1]{\bm{#1}^\star}
\newcommand{\tildebm}[1]{\widetilde{\bm{#1}}}

\newcommand{\preprintswitch}[2]{#2} % For Preprint

\newtheorem{lemma}{Lemma}
\newtheorem{proposition}{Proposition}

\newtheorem{definition}{Definition} 
\newtheorem{remark}{Remark}

\newtheorem{assumption}{Assumption}
\begin{document}
\title{
	 \LARGE A Behavioral Input-Output Parametrization of Control Policies with Suboptimality Guarantees
	}

	 \author{Luca Furieri, Baiwei Guo\IEEEauthorrefmark{1}, Andrea Martin\IEEEauthorrefmark{1}, and Giancarlo Ferrari-Trecate
    
    \thanks{Authors are with the Institute of Mechanical Engineering, École Polytechnique Fédérale de Lausanne, Switzerland. E-mails: {\tt\footnotesize \{luca.furieri, baiwei.guo, andrea.martin, giancarlo.ferraritrecate\}@epfl.ch}} % <-this % stops a space % <-this % stops a space

   %\thanks{\IEEEauthorrefmark{2}Andrea Martin is also with the Automatic Control Laboratory, Department of Information Technology and Electrical Engineering, ETH Z\"urich, Switzerland.}
    
    \thanks{\IEEEauthorrefmark{1}Baiwei Guo and Andrea Martin contributed equally to this work.} 
    
    \thanks{Research supported by the Swiss National Science Foundation under the NCCR Automation (grant agreement 51NF40\textunderscore 80545).} % 
    }

\maketitle
	\allowdisplaybreaks
	
	\IEEEpeerreviewmaketitle

\begin{abstract}

Recent work in data-driven control has revived behavioral theory to perform a variety of complex control tasks, by directly plugging libraries of past input-output trajectories into optimal control problems. Despite recent advances, a key aspect remains unclear: how and to what extent do noise-corrupted data impact control performance? In this work, we provide a quantitative answer to this question. We formulate a Behavioral version of the Input-Output Parametrization (BIOP) for the optimal predictive control of unknown systems using output-feedback dynamic control policies. 
The main advantages of the proposed framework are that 1) the state-space parameters and the initial state need not be specified for controller synthesis, 2) it can be used in combination with state-of-the-art impulse response estimators, 
and 3) it allows to recover suboptimality results on learning the Linear Quadratic Gaussian (LQG) controller, therefore revealing, in a quantitative way, how  the level of noise in the data affects the performance of behavioral methods. Specifically, it is shown that the performance degrades linearly with the prediction error of the behavioral model. We conclude the paper with numerical experiments to validate our results.

\end{abstract}
\section{Introduction}
Several safety-critical engineering systems that play a crucial role in our modern society are becoming too complex to be accurately modeled through white-box models \cite{lamnabhi2017systems}. As a consequence, most modern control perspectives envision unknown black-box systems for which an optimal behavior must be attained by solely relying on a collection of historical system's output trajectories in response to different inputs.

Widely speaking, we can design optimal controllers from data according to two paradigms. The first category contains \emph{model-based} methods, where historical input-output trajectories are exploited to approximate the system parameters, and a suitable controller is computed for this estimated model. The second category contains \emph{model-free} methods, where one aims to learn the best control policy directly by observing historical trajectories, without explicitly reconstructing an internal representation of the dynamical system. Both approaches possess their own potential and limitations; among numerous recent surveys, we refer to  \cite{recht2019tour}.

Given the intricacy of establishing rigorous suboptimality and sample-complexity bounds, most recent model-based and model-free approaches have focused on basic Linear Quadratic Regulator (LQR) and Linear Quadratic Gaussian (LQG) control problems as suitable benchmarks to establish how machine learning  can be interfaced to the continuous action spaces typical of control~\cite{dean2019sample, fazel2018global,malik2018derivative,zhengfurieri2020sample,simchowitz2020improper,lale2020logarithmic,tsiamis2020sample,zhang2020policy}. When it comes to complex tasks, such as constrained and distributed control, it is more challenging to perform a rigorous probabilistic analysis. Recent advances include \cite{dean2019safely,fattahi2020efficient} for constrained and distributed LQR control with direct state measurements and \cite{furieri2020learning} for distributed output-feedback LQG.

A promising data-driven approach that aims at bypassing a parametric description of the system dynamics, while still being conceptually simple to implement for the users, hinges on the \emph{behavioral framework} \cite{willems1997introduction}. This approach has gained renewed interest with the introduction of Data-EnablEd Predictive Control \cite{coulson2019data,coulson2020distributionally,dorfler2021bridging}, which established that constrained output reference tracking can be effectively tackled in a Model-Predictive-Control (MPC) fashion by plugging adequately generated historical data into a convex optimization problem. 
In parallel, \cite{de2019formulas} introduced data-driven formulations for some controller design tasks. These works inspired several extensions including closed-loop control with stability guarantees \cite{berberich2020data}, maximum-likelihood identification for control \cite{iannelli2020experiment,yin2020maximum}, and nonlinear variants \cite{lian2021nonlinear}.

In practice, however, historical data are corrupted by noise and the quality and coherency of the achieved solutions may be compromised. While several approaches have recently been proposed,  e.g. \cite{coulson2020distributionally,alpago2020extended,yin2020maximum}, a complete quantitative analysis for the noisy case is still unavailable.  Recently, \cite{xue2020data} has derived suboptimality and sample-complexity bounds through a data-driven formulation of the System Level Synthesis (SLS) approach. However, a limiting assumption in \cite{xue2020data} is that the internal system states can be measured directly. 

%%%GOAL
Our main contribution is to propose a behavioral optimal control framework for partially observed systems. Specifically, we leverage recent Input-Output Parametrization (IOP) tools \cite{furieri2019input} for optimal output-feedback controller design and set up a data-driven  formulation built upon behavioral theory; we denote the resulting framework as Behavioral IOP (BIOP). The advantages of the proposed BIOP are threefold. First, it solely relies on libraries of past input-output trajectories, therefore enabling optimal controller synthesis without specifying the system's state-space parameters and the value of the state at the initial time. Second, the system impulse response is replaced by a suitable linear combination of historical noisy input-output trajectories, which may encompass, for instance, standard least-squares solutions \cite{oymak2019non}, data-enabled Kalman filtering \cite{alpago2020extended}, and the recently proposed signal matrix models (SMM) \cite{yin2020maximum,iannelli2020experiment}. Third, our framework allows one to quantify the incurred suboptimality as a function of the level of the noise  corrupting the available data; this is achieved  by first establishing a tractable method to synthesize robust BIOP controllers and then adapting recent results from  \cite{zhengfurieri2020sample}. As a further contribution, we include the effect of a non-zero noisy initial condition in the analysis. 
To the best of our knowledge, noise-dependent suboptimality guarantees on using behavioral theory for output-feedback control have not been established before.

We formulate the control problem in Section~\ref{sec:prob_statement}. Section~\ref{sec:BIOP_nonnoisy} derives the BIOP, a data-driven version of the IOP valid when the data are noiseless. Section~\ref{sec:BIOP_noisy} establishes a tractable robust version of the BIOP which can be used when the data are noisy. Section~\ref{sec:suboptimality} formally quantifies the suboptimality incurred by the solution of the robust BIOP. We present numerical experiments in Section~\ref{sec:experiments} for validating our results and we conclude the paper in Section~\ref{sec:conclusions}. \preprintswitch{Proofs are provided in the appendices of the accompanying Arxiv report \cite{furieri2021behavioral}.}{} 
\subsection{Notation}
We use $\mathbb{R}$ and  $\mathbb{N}$ to denote real numbers and non-negative integers, respectively. We use $I_n$ to denote the identity matrix of size $n \times n$ and  $0_{m \times n}$ to denote the zero matrix of size $m \times n$. We write $M=\text{blkdg}(M_1,\ldots,M_N)$ to denote a block-diagonal matrix with $M_1,\ldots,M_N \in \mathbb{R}^{m \times n}$ on its diagonal block entries, and for $\mathbf{M} =  \begin{bmatrix}M_1^\mathsf{T}&\cdots&M_N^\mathsf{T}\end{bmatrix}^\mathsf{T}$ we define the block-Toeplitz matrix $$\text{Toep}_{m \times n}\left(\mathbf{M}\right) \hspace{-0.075cm} = \hspace{-0.075cm} \begin{bmatrix}
           M_1&0_{m \times n}&\cdots&0_{m \times n}\\ M_2&M_1&\cdots &0_{m \times n}\\
           \vdots&\vdots&\ddots&\vdots\\
           M_N&M_{N-1}&\cdots&M_1
    \end{bmatrix}
    %\,
    .$$
More concisely, we will write $\text{Toep}(\cdot)$ when the dimensions of the blocks are clear from the context. The Kronecker product between $M \in \mathbb{R}^{m \times n}$ and $P \in \mathbb{R}^{p \times q}$ is denoted as $M \otimes P \in \mathbb{R}^{mp \times nq}$. Given $K \in \mathbb{R}^{m \times n}$,  $\text{vec}(K) \in \mathbb{R}^{mn}$ is a column vector that stacks the columns of $K$. The Euclidean norm of a vector $v \in \mathbb{R}^n$ is denoted by $\norm{v}_2^2=v^\mathsf{T}v$ and the induced two-norm of a matrix $M \in \mathbb{R}^{m \times n}$ is defined as $\sup_{\norm{x}_2 = 1} \norm{Mx}_2$. The Frobenius norm of a matrix $M \in \mathbb{R}^{m \times n}$ is denoted by $\norm{M}_{F}=\sqrt{\text{Trace}(M^\mathsf{T}M)}$. For a symmetric matrix $M$, 
we write $M \succ 0$ (resp. $M \succeq 0$) if and only if it is positive definite (resp. positive semidefinite). We say that $x\sim \mathcal{N}(\mu,\Sigma)$ if the random variable  $x\in \mathbb{R}^n$ is distributed according to a normal distribution with mean $\mu \in \mathbb{R}^n$ and covariance matrix $\Sigma \succeq 0$ with $\Sigma \in \mathbb{R}^{n \times n}$.

A \emph{finite-horizon} trajectory of length $T$ is a sequence $\omega(0),\omega(1),\cdots, \omega(T-1)$ with $\omega(t) \in \mathbb{R}^n$ for every $t=0,1,\ldots, T-1$, which can be compactly written as
\begin{equation*}
    \bm{\omega}_{[0,T-1]} = \begin{bmatrix}\omega^\mathsf{T}(0)&\omega^\mathsf{T}(1)& \ldots&\omega^\mathsf{T}(T-1)\end{bmatrix}^\mathsf{T} \in \mathbb{R}^{nT}\,.
\end{equation*}
When the value of $T$ is clear from the context, we will omit the subscript $[0,T-1]$. For a  finite-horizon trajectory $\bm{\omega}_{[0,T-1]}$ we also define the Hankel matrix of depth $L$ as
\begin{equation*}
    \mathcal{H}_{L}(\bm{\omega}_{[0,T-1]}) = \begin{bmatrix}
           \omega(0)&\omega (1)&\cdots&\omega(T-L)\\ \omega(1)&\omega(2)&\cdots &\omega(T-L+1)\\
           \vdots&\vdots&\ddots&\vdots\\
           \omega(L-1)&\omega(L)&\cdots&\omega(T-1)
    \end{bmatrix}\,.
\end{equation*}

\section{Problem Statement}
\label{sec:prob_statement}
We consider a linear system with output observations, whose state-space  representation is given by
%\vspace{-3pt}
\begin{equation} \label{eq:dynamic}
	\begin{aligned}
	    x(t+1) = A x(t)+B u(t),~~y(t) = C x(t) + v(t)\,,
	\end{aligned}    
\end{equation}
where $x(t) \in \mathbb{R}^n$ is the state of the system and $x(0) =x_0$ for a predefined $x_0 \in \mathbb{R}^n$, $u(t)\in \mathbb{R}^m$ is the control input, $y(t)\in \mathbb{R}^p$ is the observed output, and $v(t)\in \mathbb{R}^p$ denotes Gaussian measurement noise $v(t) \sim \mathcal{N}(0,\Sigma_v)$, with $\Sigma_v \succ 0$. The system is controlled through a time-varying, dynamic linear control policy of the form
\begin{equation}
\label{eq:input}
	    u(t) = \sum_{k=0}^t K_{t,k} y(k)+w(t)\,,
	    \end{equation}
where  $w(t)\in \mathbb{R}^m$ denotes Gaussian noise on the input  $w(t) \sim \mathcal{N}(0,\Sigma_w)$ with $\Sigma_w \succeq 0$. 
Similar to standard LQG, our control goal is to synthesize a feedback control policy that minimizes the expected value with respect to the disturbances of a quadratic objective defined over future input-output trajectories for a horizon $N \in \mathbb{N}$:
    \begin{equation}
\label{eq:cost_output}
J^2:=\mathbb{E}_{w,v}\left[\sum_{t=0}^{N-1}\left(y(t)^\mathsf{T}L_ty(t)+u(t)^\mathsf{T}R_tu(t)\right)\right]\,,
\end{equation}	
where $L_t\succeq0$, $R_t\succ 0$ for every $t = 0,\cdots,N-1$. We note that, with Gaussian noise, dynamic linear policies are optimal for the cost defined in \eqref{eq:cost_output}. 
\begin{remark}
The reader might have noticed that the problem of minimizing \eqref{eq:cost_output} for a system in the form \eqref{eq:dynamic}-\eqref{eq:input} is slightly different from some of the classical LQG formulations, see for instance  \cite{bertsekas2011dynamic}. Specifically, in \eqref{eq:cost_output} we penalize the \emph{outputs} instead of the states, and the input noise $w(t)$ enters the state equation indirectly through the matrix $B$. This choice is motivated as follows:
\begin{enumerate}
    \item For all practical purposes, the cost function must be defined by the user. In a data-driven setup where only input-output samples can be measured, the user has to evaluate the cost solely relying on input-output trajectories. Furthermore, to define the cost, it is natural to specify the variance of the noise affecting inputs and outputs; instead it would be less meaningful to specify the statistics of the noise entering the states, as these would be representation dependent (i.e. only specified up to a change of variables $z(t) = S x(t)$, where $S$ is unknown because we do not have access to $x(t)$ by assumption).
    \item For zero initial state, the system \eqref{eq:dynamic} is equivalent to a classical transfer function representation as per Figure~\ref{Fig:LTI}. The considered noise model is indeed the standard choice in  closed-loop plant $\mathcal{H}_2$ norm minimization, see for instance \cite{zhou1996robust}. 
\end{enumerate}
\end{remark}

\begin{remark}
In this work, we focus on solving and analyzing a  finite-horizon control problem, which represents one iteration of a receding-horizon Model Predictive Control (MPC) implementation scheme. It is therefore appropriate to compare the proposed approach with a \emph{single} iteration of the DeePC setup in \cite{coulson2019data,berberich2020data}. The main difference is that we perform \emph{closed-loop predictions}, i.e., we optimize over feedback policies $\pi(\cdot)$ such that $u(t)=\pi(y(t),\ldots,y(0))$, while the DeePC \cite{coulson2019data,berberich2020data} performs \emph{open-loop predictions}, i.e., it directly optimizes over input sequences  $u(0),u(1),u(N-1)$. For linear systems subject to polytopic safety constraints, it is well-known that closed-loop predictions are less conservative than open-loop ones and allow for longer prediction horizons without incurring in infeasibility \cite{bemporad1998reducing}. The price to pay for such performance improvement is an increased computational burden due to the larger dimensionality of the problem.  
\end{remark}

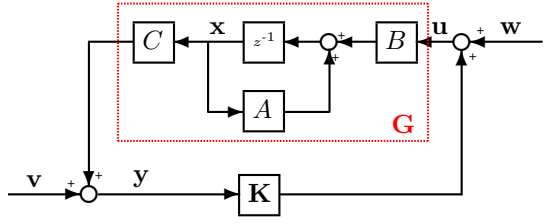
\begin{figure}[t]
  \centering
  \begin{center}\input{interconnection.tex}\end{center}
  \caption{Interconnection of the plant $\mathbf{G}$ and the controller $\mathbf{K}$, where $z^{-1}$ denotes the standard time-shift operator.}\label{Fig:LTI}
\end{figure}

\subsection{Strongly convex design through the IOP}
By leveraging tools offered by the framework of the IOP \cite{furieri2019input},  we formulate a strongly convex program that computes the optimal feedback control policy by finding the optimal input-output closed-loop responses. The state-space equations \eqref{eq:dynamic} provide the following relations between trajectories
	\begin{align}
	%\label{eq:compact}
	&\mathbf{x}_{[0,N-1]}=\mathbf{P}_{A}(:,0)x(0)+\mathbf{P}_{B}\mathbf{u}_{[0,N-1]}\,,\label{eq:state_compact}\\ &\mathbf{y}_{[0,N-1]}=\mathbf{C}\mathbf{x}_{[0,N-1]}+\mathbf{v}_{[0,N-1]}\,,  \label{eq:output_compact}
	\end{align}
	where  $\mathbf{P}_A(:,0)$ denotes the first block-column of $\mathbf{P}_A$ and
	\begin{alignat*}{3}
	&\mathbf{P}_{A}=(I-\mathbf{Z}\mathbf{A})^{-1}\,, &&\quad  \mathbf{P}_{B}=(I-\mathbf{Z}\mathbf{A})^{-1}\mathbf{Z}\mathbf{B}\,,\\
	& \mathbf{A}=I_{N} \otimes A\,, && \quad \mathbf{B}=I_{N}\otimes B\,,\\
& \mathbf{C} = I_{N} \otimes C\,, && \quad \mathbf{Z}=\begin{bmatrix}
0_{n \times n(N-1)}&0_{n \times n}\\
I_{n(N-1)}&0_{n(N-1) \times n}
\end{bmatrix}\,.
	\end{alignat*}
We note that $\mathbf{CP}_B$ is a Toeplitz matrix with blocks in the form $CA^iB$. From now on, we equivalently denote $\mathbf{G} = \mathbf{CP}_B$ to highlight that $\mathbf{G}$ is a block-Toeplitz matrix containing the first $N$ components of the impulse response of the plant $\mathbf{G}(z) = C(zI-A)^{-1}B$ reported in Figure~\ref{Fig:LTI}.  Second, with similar reasoning, the matrix $\mathbf{CP}_{A}(:,0)$ contains the observability terms $CA^i$ for $i=0,\ldots,N-1$.
The control policy can be rewritten as:
	\begin{equation}
	\label{eq:control_policy}
    \mathbf{u}_{[0,N-1]} = \mathbf{K}\mathbf{y}_{[0,N-1]}+\mathbf{w}_{[0,N-1]}\,,
\end{equation}
where $\mathbf{K}$ has a causal sparsity pattern:
\begin{equation}
\label{eq:K_sparsity}
\mathbf{K}=\begin{bmatrix}
	K_{0,0}&0_{m \times p}&\cdots&0_{m \times p}\\
	K_{1,0}&K_{1,1}&\ddots&0_{m \times p}\\
	\vdots&\vdots&\ddots&\vdots\\
	K_{N-1,0}&K_{N-1,1}&\cdots&K_{N-1,N-1}%\\
%	0_{m \times p}&0_{m \times p}&\cdots&0_{m \times p}&0_{m \times p}
	\end{bmatrix}\,.
\end{equation}
By plugging the controller \eqref{eq:control_policy} into \eqref{eq:state_compact}-\eqref{eq:output_compact}, it is easy to derive the relationships
\begin{align}
    &\begin{bmatrix}\mathbf{y}\\ \mathbf{u}\end{bmatrix}  =\begin{bmatrix}
           \bm{\Phi}_{yy} & \bm{\Phi}_{yu} \\
            \bm{\Phi}_{uy} & \bm{\Phi}_{uu} 
        \end{bmatrix}\begin{bmatrix} \mathbf{v}+\mathbf{CP}_{A}(:,0)x(0)\\ \mathbf{w}\end{bmatrix}\,, \label{eq:IOP_parameters}
\end{align}
where
\begin{equation}
   \begin{bmatrix}
           \bm{\Phi}_{yy} & \bm{\Phi}_{yu} \\
            \bm{\Phi}_{uy} & \bm{\Phi}_{uu} 
        \end{bmatrix} = \begin{bmatrix}(I-\mathbf{GK})^{-1} & (I-\mathbf{GK})^{-1}\mathbf{G}\\ \mathbf{K}(I-\mathbf{GK})^{-1} & (I-\mathbf{KG})^{-1}\end{bmatrix}\label{eq:CL_responses}\,.
\end{equation}
The parameters ($\bm{\Phi}_{yy}, \bm{\Phi}_{yu}, \bm{\Phi}_{uy}, \bm{\Phi}_{uu}$)  represent the four closed-loop responses defining the relationship between disturbances and input-output signals.  The main concept behind the IOP in \cite{furieri2019input} is that linear output-feedback control policies $\mathbf{K}$ can be expressed in terms of corresponding closed-loop responses that lie in an affine subspace, hence enabling a convex formulation of the objective  $J(\mathbf{G},\mathbf{K})$ given in \eqref{eq:cost_output} as a function of the closed-loop responses. The IOP serves well our purposes in a data-driven output-feedback setup, as it offers a controller parametrization that is directly defined through the impulse response parameters $\mathbf{G}$, without requiring a state-space representation. We adapt the following result from \cite{furieri2019input} to the finite horizon case. A proof is reported in the Appendix \preprintswitch{of \cite{furieri2021behavioral}}{for completeness}.
\begin{proposition}
\label{prop:IOP}
Consider the LTI system \eqref{eq:dynamic} evolving under the control policy \eqref{eq:control_policy} within a finite horizon of length $N \in \mathbb{N}$. Then:
\begin{enumerate}
    \item For any controller $\mathbf{K}$ there exist four matrices ($\bm{\Phi}_{yy}, \bm{\Phi}_{yu}, \bm{\Phi}_{uy}, \bm{\Phi}_{uu}$) such that $\mathbf{K} = \bm{\Phi}_{uy}\bm{\Phi}_{yy}^{-1}$ and
    \begin{align}
        \begin{bmatrix} I & -\mathbf{G} \end{bmatrix}\begin{bmatrix}
               \bm{\Phi}_{yy} & \bm{\Phi}_{yu} \\
                \bm{\Phi}_{uy} & \bm{\Phi}_{uu} 
            \end{bmatrix} &= \begin{bmatrix} I & 0 \end{bmatrix}, \label{eq:ach1}\\
                \begin{bmatrix}
                \bm{\Phi}_{yy} & \bm{\Phi}_{yu} \\
                \bm{\Phi}_{uy} & \bm{\Phi}_{uu} 
            \end{bmatrix}\begin{bmatrix}  -\mathbf{G} \\I \end{bmatrix} &= \begin{bmatrix} 0 \\ I\end{bmatrix},\label{eq:ach2} \\
           \bm{\Phi}_{yy}, \bm{\Phi}_{uy}, 
                \bm{\Phi}_{yu}, \bm{\Phi}_{uu} %&
                \text{ have} &\text{ causal sparsities \footnotemark}.\label{eq:ach3}
    \end{align}
    \footnotetext{Specifically, they have the block lower-triangular sparsities resulting by construction from the expressions \eqref{eq:CL_responses}, the sparsity of $\mathbf{K}$ in \eqref{eq:K_sparsity} and that  of $\mathbf{G}$.}
    \item For any four matrices ($\bm{\Phi}_{yy}, \bm{\Phi}_{yu}, \bm{\Phi}_{uy}, \bm{\Phi}_{uu}$) lying in the affine subspace \eqref{eq:ach1}-\eqref{eq:ach3}, the controller $\mathbf{K}=\bm{\Phi}_{uy}\bm{\Phi}_{yy}^{-1}$ is causal as per \eqref{eq:K_sparsity} and yields the closed-loop responses  ($\bm{\Phi}_{yy}, \bm{\Phi}_{yu}, \bm{\Phi}_{uy}, \bm{\Phi}_{uu}$).
\end{enumerate}
\end{proposition}

We are now ready to establish a strongly convex formulation of the optimal control problem under study. Please refer to the \preprintswitch{Appendix of \cite{furieri2021behavioral}}{Appendix} for a complete proof.
\begin{proposition}
\label{prop:strongly_convex}
Consider the LTI system \eqref{eq:dynamic}. The controller in the form \eqref{eq:control_policy} achieving the minimum of the cost functional \eqref{eq:cost_output} is given by $\mathbf{K} = \bm{\Phi}_{uy} \bm{\Phi}_{yy}^{-1}$, where $\bm{\Phi}_{uy},\bm{\Phi}_{yy}$ are optimal solutions to the following strongly convex program:
\begin{align}
    &~\min_{\bm{\Phi}}\frobenius{
    \hspace{-0.025cm}
    \begin{bmatrix}\mathbf{L}^{\frac{1}{2}}&0\\0&\mathbf{R}^{\frac{1}{2}}\end{bmatrix}
    \hspace{-0.1cm}
    \begin{bmatrix}
            \bm{\Phi}_{yy} & \bm{\Phi}_{yu} \\
            \bm{\Phi}_{uy} & \bm{\Phi}_{uu} 
        \end{bmatrix}
        \hspace{-0.1cm}
        \begin{bmatrix}\bm{\Sigma}^{\frac{1}{2}}_v&0&\mathbf{y}_{x(0)}\\0&\bm{\Sigma}^{\frac{1}{2}}_w&0\end{bmatrix}
        \hspace{-0.025cm}
        }^2\label{prob:IOP}\\
    &\st~ \eqref{eq:ach1}-\eqref{eq:ach3}\,, \nonumber
\end{align}
where $\mathbf{y}_{x(0)}=\mathbf{CP}_{A}(:,0)x(0)$, $\mathbf{L} = \text{\emph{blkdiag}}(L_0,\cdots, L_{N-1})$, $\mathbf{R} = \text{\emph{blkdiag}}(R_0,\cdots, R_{N-1})$, $\bm{\Sigma}_v = I_{N} \otimes \Sigma_v$ and $\bm{\Sigma}_w = I_{N} \otimes \Sigma_w$.
\end{proposition}

When the system parameters $(A,B,C,x_0)$ are known, it is straightforward and efficient to compute the unique globally optimal solution $(\bm{\Phi}^\star_{yy}, \bm{\Phi}^\star_{yu}, \bm{\Phi}^\star_{uy}, \bm{\Phi}^\star_{uu}$) of problem~\eqref{prob:IOP} with off-the-shelf interior point solvers. The globally optimal control policy is recovered as $\mathbf{K}^\star = \bm{\Phi}_{uy}^\star (\bm{\Phi}^\star_{yy})^{-1}$.  We also remark that, since the noise is Gaussian, the linear policy $\mathbf{u}=\pi^\star(\mathbf{y}) = \mathbf{K}^\star \mathbf{y}$ is optimal with respect to all feedback policies. If the noise is non-Gaussian, $\mathbf{K}^\star$ remains the optimal linear controller, but nonlinear policies may outperform it.

However, it is more challenging to compute $\mathbf{K}^\star$ merely relying on libraries of past input-output trajectories. In the next section, we exploit behavioral theory to provide a non-parametric data-driven version of \eqref{prob:IOP}.

\section{Behavioral Input-Output Parametrization}
\label{sec:BIOP_nonnoisy}

Before moving on, we  recall the following definition of persistency of excitation and the result known as the \emph{Fundamental Lemma} for LTI systems \cite{willems2005note}.
\begin{definition}
We say that $\mathbf{u}^h_{[0,T-1]}$ is \emph{persistently exciting} (PE) of order $L$ if the Hankel matrix $\mathcal{H}_L(\mathbf{u}^h_{[0,T-1]})$ has full row-rank.
\end{definition}
A necessary condition for the matrix $\mathcal{H}_{L}(\mathbf{u}^h_{[0,T-1]})$ to be full row-rank is that it has at least as many columns as rows. It follows that the trajectory $\mathbf{u}^h_{[0,T-1]}$  must be long enough to satisfy $T\geq (m+1)L-1$.
\begin{lemma}[Theorem 3.7, \cite{willems2005note}]
\label{le:Willems}
Consider system \eqref{eq:dynamic} and assume that $(A,B)$ is controllable and that there is no noise. Let $\{\mathbf{y}^h_{[0,T-1]},\mathbf{u}^h_{[0,T-1]}\}$ be a historical system trajectory of length $T$. Then, if $\mathbf{u}_{[0,T-1]}$ is PE of order $n+L$, the signals $\mathbf{y}^\star_{[0,L-1]}\in \mathbb{R}^{pL}$ and $\mathbf{u}^\star_{[0,L-1]} \in \mathbb{R}^{mL}$ are valid trajectories of \eqref{eq:dynamic} if and only if there exists $g \in \mathbb{R}^{T-L+1}$ such that
\begin{equation}
    \begin{bmatrix}
           \mathcal{H}_{L}(\mathbf{y}^h_{[0,T-1]})\\
           \mathcal{H}_{L}(\mathbf{u}^h_{[0,T-1]})
    \end{bmatrix} g = \begin{bmatrix}\mathbf{y}^\star_{[0,L-1]}\\\mathbf{u}^\star_{[0,L-1]}\end{bmatrix}\,. \label{eq:Willems}
\end{equation}
\end{lemma}

Next, we  show how Lemma~\ref{le:Willems} can be directly exploited to obtain a non-parametric formulation of \eqref{prob:IOP}. We work under the following assumptions.

\begin{assumption}
\label{ass:1}
The data-generating LTI system \eqref{eq:dynamic} is such that $(A,B)$ is controllable and $(A,C)$ is observable.
\end{assumption}

\begin{assumption}
\label{ass:2}
The following data are available:
    \begin{enumerate}
    \item[i)] a \emph{recent} system trajectory of length $T_{ini}$: $\left\{\mathbf{y}^r_{[0,T_{ini}-1]},\mathbf{u}^r_{[0,T_{ini}-1]}\right\}$, with $\mathbf{y}^r_{[0,T_{ini}-1]} \hspace{-0.0175cm} = \hspace{-0.0175cm} \mathbf{y}_{[-T_{ini},-1]}$ and $\mathbf{u}^r_{[0,T_{ini}-1]} = \mathbf{u}_{[-T_{ini},-1]}$,% corresponding to the trajectory in the immediate past that brought the system to its current initial state $x(0)$.
        \item[ii)] a \emph{historical} system trajectory of length $T$: $        \left\{\mathbf{y}^h_{[0,T-1]},\mathbf{u}^h_{[0,T-1]}\right\}$, with $\mathbf{y}^h_{[0,T-1]} = \mathbf{y}_{[-T_{h},-T_{h}+T-1]}$ and $\mathbf{u}^h_{[0,T-1]} = \mathbf{u}_{[-T_{h},-T_h+T-1]}$ for $T_h \in \mathbb{N}$ such that $T_h > T+T_{ini}$.
    \end{enumerate}
\end{assumption}

\begin{assumption}
\label{ass:3}
The historical and recent data are not corrupted by noise. 
\end{assumption}

We will drop Assumption~3 in Section~\ref{sec:BIOP_noisy}. 

\begin{assumption}
\label{ass:4}
The historical input trajectory $\mathbf{u}^h_{[0,T-1]}$ is persistently exciting of order $n+T_{ini}+N$, where $T_{ini}\geq l$ and $l$ is the smallest integer such that \begin{equation*}\begin{bmatrix}C^\mathsf{T}&(CA)^\mathsf{T}&\cdots &(CA^{l-1})^\mathsf{T}\end{bmatrix}^\mathsf{T}\,,\end{equation*} has full row-rank. Note that if Assumption~\ref{ass:1} holds, then $l \leq n$.
\end{assumption}

A few comments are in order. First, Assumption~\ref{ass:1} is without loss of generality, as from an input-output perspective we are not concerned with the non-controllable and non-observable subsystems. Therefore, it is equivalent to assume that $(A,B,C)$ are the matrices associated with the controllable and observable parts of the LTI system. Second, in Assumption~\ref{ass:2} the \emph{historical} data are needed to construct a non-parametric system representation, and the \emph{recent} data are exploited to define a cost function that accurately reflects the system initial state $x(0) \in \mathbb{R}^n$.  Third, in Assumption~\ref{ass:3} we assume that the observed data are noiseless to construct a data-driven optimal control problem  that is \emph{equivalent} to \eqref{prob:IOP}. We will deal with the noisy case in Section~\ref{sec:BIOP_noisy}.
\begin{theorem}[Behavioral IOP]
\label{th:BIOP}
Consider the unknown LTI system \eqref{eq:dynamic} and let Assumptions~\ref{ass:1}-\ref{ass:4} hold. Let $(G,g)$ be any solutions to the linear system of equations
\begin{equation}
\label{eq:BM}
    \begin{bmatrix}U_p\\ Y_p\\U_f\end{bmatrix}
    \hspace{-0.1cm}
    \begin{bmatrix} G&g \end{bmatrix}
    \hspace{-0.06cm}
    =
    \hspace{-0.06cm}
    \begin{bmatrix}0_{mT_{ini} \times m}&\mathbf{u}^r_{[0,T_{ini}-1]}\\0_{pT_{ini} \times m} & \mathbf{y}^r_{[0,T_{ini}-1]}\\\begin{bmatrix}I_m&0_{ m \times m(N-1)} \end{bmatrix}^\mathsf{T} &0_{mN \times 1}\end{bmatrix}\hspace{-0.1cm},%\,,
\end{equation}
where $\begin{bmatrix}U_p\\U_f\end{bmatrix} = \mathcal{H}_{T_{ini}+N}(\mathbf{u}^h_{[0,T-1]})$ and   $\begin{bmatrix}Y_p\\Y_f\end{bmatrix} = \mathcal{H}_{T_{ini}+N}(\mathbf{y}^h_{[0,T-1]})$. Then, the optimization problem~\eqref{prob:IOP} is equivalent to \begin{align}
    &\min_{\bm{\Phi}}\frobenius{\begin{bmatrix}\mathbf{L}^{\frac{1}{2}}&0\\0&\mathbf{R}^{\frac{1}{2}}\end{bmatrix}\begin{bmatrix}
            \bm{\Phi}_{yy} & \bm{\Phi}_{yu} \\
            \bm{\Phi}_{uy} & \bm{\Phi}_{uu} 
        \end{bmatrix}\begin{bmatrix}\bm{\Sigma}^{\frac{1}{2}}_v&0&Y_fg\\0&\bm{\Sigma}^{\frac{1}{2}}_w&0\end{bmatrix}}^2\label{prob:IOP_data}\\
    &\st \begin{bmatrix} I & -\text{\emph{Toep}}(Y_fG) \end{bmatrix}\begin{bmatrix}
           \bm{\Phi}_{yy} & \bm{\Phi}_{yu} \\
            \bm{\Phi}_{uy} & \bm{\Phi}_{uu} 
        \end{bmatrix} = \begin{bmatrix} I & 0 \end{bmatrix}\,, \nonumber\\
            &\qquad\quad\begin{bmatrix}
            \bm{\Phi}_{yy} & \bm{\Phi}_{yu} \\
            \bm{\Phi}_{uy} & \bm{\Phi}_{uu} 
        \end{bmatrix}\begin{bmatrix} -\text{\emph{Toep}}(Y_fG) \\I \end{bmatrix} = \begin{bmatrix} 0 \\ I\end{bmatrix} \,,\nonumber\\
       &\qquad\quad\bm{\Phi}_{yy}, \bm{\Phi}_{uy}, 
            \bm{\Phi}_{yu}, \bm{\Phi}_{uu}  \text{\emph{ with causal sparsities}}.\nonumber
            %&~&&\mathbf{y}_{free} = %\mathbf{y}^r_{[0,N]}-Y_fG\mathbf{u}^r_{[0,N]}\,. \label{eq:free_response}
\end{align}
\end{theorem}
\begin{proof}
 In problem~\eqref{prob:IOP}, the system parameters $(A,B,C,x(0))$ appear through the terms $\mathbf{G} =\mathbf{CP}_B$ in the constraints and $\mathbf{CP}_Ax(0)$ in the cost. It is therefore sufficient to show that we are able to substitute both elements with data as per the theorem statement.
 
 Let $G$ be any solution \eqref{eq:BM}. By rearranging the terms, each column of $G$ can be thought as a solution to \eqref{eq:Willems} associated with a zero initial condition and a unitary input $e_i \in \mathbb{R}^m$. Since the hypotheses of Lemma~\ref{le:Willems} are satisfied for $L = T_{ini}+N$, similar to Proposition~11 of \cite{markovsky2008data} we deduce that $Y_f G$ is the first block-column of the system impulse response matrix, independent of the solution $G$. Therefore, we can equivalently substitute  $\mathbf{G} = \text{Toep}_{p \times m}(Y_fG)$ in the constraints \eqref{eq:ach1}-\eqref{eq:ach2} of problem~\eqref{prob:IOP}. Finally, note that $Y_fg$ corresponds to the trajectory starting at $x(0)$ (as implicitly defined by the recent trajectory $\mathbf{y}_{[-T_{ini},-1]}$ and $\mathbf{u}_{[-T_{ini},-1]}$) when applying a zero input  \cite{markovsky2008data}. Therefore, it corresponds to the true free response starting from $x(0)$. 
\end{proof}

For any solution $G$ of the behavioral impulse response  representation \eqref{eq:BM}, the affine constraints \eqref{eq:ach1}-\eqref{eq:ach3} describe all the achievable closed-loop responses for the unknown model and the corresponding controller $\mathbf{K}$. Also, for any solution $g$ of \eqref{eq:BM}, the term $Y_f g$ represents the true free response of the system. As a result, the achieved optimal controller $\mathbf{K}^\star$ and optimal cost $J^\star$ are independent of the chosen solution $(G,g)$ for \eqref{eq:BM}. We have thus characterized a data-driven version of the IOP. Theorem~\ref{th:BIOP} further shows that, by exploiting the BIOP, it is straightforward to cast the LQG problem as a strongly convex program. 

\begin{remark}
To use the language of \cite{dorfler2021bridging,coulson2019data,iannelli2020experiment}, the proposed BIOP formulation belongs to the class of \emph{indirect}, \emph{non-parametric} data-driven controller synthesis methods enabled by behavioral theory. Indeed, the optimal feedback controller is computed in two phases, hence the adjective \emph{indirect}. First,  an impulse response matrix is obtained as part of an implicit identification step based on Willems's fundamental lemma. Second, an optimal control problem is cast and solved by replacing the impulse and free responses with a suitable linear combinations of historical input-output trajectories . The works in \cite{de2019formulas,xue2020data}, propose an alternative  \emph{direct} approach where a single, high-dimensional optimization problem is solved; the decision variables are the weights to be assigned to the different columns of the data Hankel matrix rather than the system closed-loop responses.

A thorough analysis of the advantages and disadvantages inherent to direct or indirect behavioral approaches is a topic of ongoing research in the field. Here, we note  a few initial  points. First, the proposed indirect BIOP can directly encapsulate recent results on statistically optimal non-parametric estimation of an impulse response matrix \cite{yin2020maximum,iannelli2020experiment,alpago2020extended}.  Second,  \eqref{prob:IOP_data} involves a number of  decision variables that only scales with $N$, $m$ and $p$, while in the cost of a direct method the decision variables involved in the control cost would also scale with $T$. Last, we notice that a \emph{direct} BIOP formulation can most likely be obtained by adapting, for instance, the results of Section~VI in  \cite{de2019formulas}; we leave this topic for future work.
\end{remark}

\begin{remark}
While other parametrizations equivalent to the IOP exist, including the System Level Parametrization (SLP) \cite{wang2019system}, and other mixed parametrizations (see \cite{zheng2019system} for a survey), the IOP may be particularly well-suited for an output-feedback data-driven setup. Indeed, the SLP and the mixed parametrizations in \cite{zheng2019system} all explicitly involve state-space parameters in the constraints. By solely using input-output trajectories, the state-space parameters can only be recovered up to an unknown change of variables \cite{markovsky2006exact}, which may be problematic for defining an initial state and noise variances in the LQG cost. Instead, the BIOP is uniquely defined from data, as it only depends on the impulse response matrix without resorting to an internal state representation. 
\end{remark}

\section{Robust BIOP with Noise-Corrupted Data}
\label{sec:BIOP_noisy}
The linear system \eqref{eq:BM} is highly underdetermined when the historical trajectory is very long and noiseless. In particular, any solution $(G,g)$ to \eqref{eq:BM} gives an exact impulse response matrix and free trajectory of the system. In practice, however, the historical and recent data are corrupted by noise. According to the system equations \eqref{eq:dynamic}-\eqref{eq:input}, we can assume historical and recent trajectories are affected by noise $w^h(t),w^r(t),v^h(t),v^r(t)$ at all time instants, with expected values $\mu^h_w,\mu^r_w,\mu^h_v,\mu^r_v$ and variances $\bm{\Sigma}^h_w,\bm{\Sigma}^r_w,\bm{\Sigma}^h_v,\bm{\Sigma}^r_v$ respectively. Hence, the matrix on the left-hand-side of \eqref{eq:BM} becomes full row-rank almost surely and \eqref{eq:BM} can only yield an approximated impulse response matrix and free response. This issue is well-known in the behavioral theory literature, and several promising solutions have recently been proposed \cite{coulson2019data,coulson2020distributionally,de2019formulas,yin2020maximum,alpago2020extended}. We briefly review some of them.

Letting $\hat{U}_p,\hat{Y}_p,\hat{U}_f,\hat{Y}_f$ denote the matrices built upon noisy historical data. In order to impose a block-Toeplitz structure on the impulse response matrix, one simple solution is to choose $G$ and $g$ as 
\begin{align}
\label{eq: LS_for_G}
    G=G_{LS} &= \begin{bmatrix}\hat{U}_p\\ \hat{Y}_p\\\hat{U}_f\end{bmatrix}^+%\dagger
    \begin{bmatrix}0_{mT_{ini} \times m}\\
    0_{pT_{ini} \times m}\\
    \begin{bmatrix}I_m&0_{ m \times m(N-1)} \end{bmatrix}^\mathsf{T}
    \end{bmatrix}\,, \\
    g = g_{LS} &=  \begin{bmatrix}\hat{U}_p\\\hat{Y}_p\\\hat{U}_f\end{bmatrix}^+\begin{bmatrix}\mathbf{u}^r_{[0,T_{ini}-1]}\\\mathbf{y}^r_{[0,T_{ini}-1]}\\0_{mN \times 1}\end{bmatrix}\,,\label{eq: LS_for_g}
\end{align}
and let $\hatbf{G} = \text{Toep}\left(\widehat{Y}_fG_{LS}\right)$ and $\hatbf{y}_{free} = \widehat{Y}_fg_{LS}$ be the approximate impulse and free responses. While being simple to compute, this least-squares predictor comes without strong statistical guarantees and, for the case of the impulse response matrix, it is biased in general due to the finite-impulse-response truncation error; we refer the interest reader to \cite{sedghizadeh2018data,oymak2019non}. A data-based Kalman filter based solution to reduce the effect of noise is proposed in \cite{alpago2020extended}. Another approach is to minimize a scalar functional $f(\cdot)$ that penalizes the residuals $\Xi_y =(Y_p-\hat{Y}_p)G $ and $\xi_y =(Y_p-\hat{Y}_p)g$ \cite{coulson2020distributionally}.  A choice that reflects the maximum-likelihood interpretation of total least squares is proposed in \cite{yin2020maximum} and consists in solving the optimization problems

\begin{alignat*}{3}
    G_{ML}=&\argmin_{G} &&-log\left[p\left(\begin{bmatrix}\Xi_y\\Y_fG\end{bmatrix}|~G,Y_f\right)\right]  \\
    &\st&&\begin{bmatrix}\hat{U}_p\\ \hat{U}_f\end{bmatrix}
    G=\begin{bmatrix}
    0_{mT_{ini} \times m}\\
    \begin{bmatrix}I_m&0_{ m \times m(N-1)} \end{bmatrix}^\mathsf{T}
    \end{bmatrix}\,,\\
    g_{ML}=&  \argmin_{g} && -log\left[p\left(\begin{bmatrix}\xi_y\\Y_fg\end{bmatrix}|~g,Y_f\right)\right]\\
    &\st&&\begin{bmatrix}\hat{U}_p\\ \hat{U}_f\end{bmatrix}g=\begin{bmatrix}\mathbf{u}^r_{[0,T_{ini}-1]}\\0_{mN \times 1}\end{bmatrix}\,.
\end{alignat*}

 While the above problems are nonconvex, an iterative procedure to obtain an approximate solution is proposed in \cite{yin2020maximum}. A further refinement of the technique applied to impulse response identification is established  in \cite{iannelli2020experiment} through optimal input design. Based on the above discussion, denote the estimated impulse and free responses as $\hatbf{G} = \text{Toep}(\widehat{Y}_f G )$ and $\hatbf{y}_{free} = \widehat{Y}_f g$ respectively.  Independent of the chosen estimator, we will have that
\begin{alignat*}{3}
   &\mathbb{E}[\hatbf{G}] = M_G,\quad&& \text{Var}(\text{vec}(\hatbf{G}) )= \Sigma_G\,,\\
   &\mathbb{E}[\hatbf{y}_{free}] = \mu_y,\quad&& \text{Var}(\hatbf{y}_{free} )= \Sigma_y\,,
\end{alignat*}
where $M_G = \mathbf{G}$ and $\mu_y = \mathbf{y}_{free}$ if and only if the estimators are unbiased, and where $\Sigma_G,\Sigma_y$ are ``small'' in an appropriate sense. We thus work under the assumption that, with high-probability, the errors $\norm{\mathbf{G}-\hatbf{G}}$ and $\norm{\mathbf{y}_{free}-\hatbf{y}_{free}}$ are small; the better the predictor (i.e., smaller bias and variance), the smaller the errors. Motivated as above, we abstract from the particular identification scheme and formalize the following assumption.

\begin{assumption}
\label{ass:bounded_error}
There exist $\epsilon_G>0$ and $\epsilon_0>0$ such that, for any sequence of noisy historical and recent data, with high probability
\begin{equation*}
\norm{\mathbf{G}-\hatbf{G}}_2 = \normtwo{\bm{\Delta}} \leq \epsilon_G, ~ \norm{\mathbf{y}_{free}-\hat{\mathbf{y}}_{free}}_2 = \normtwo{\bm{\delta}_0} \leq \epsilon_0\,.
\end{equation*}
We denote $\epsilon = \max(\epsilon_G,\epsilon_0)$.
\end{assumption}

%Here remark about sample complexity and "high probability"?

After condensing the effect of noise into a single error parameter $\epsilon>0$, we are ready to leverage and adapt the analysis technique recently suggested in \cite{zhengfurieri2020sample} for infinite-horizon LQG, which follows the philosophy first introduced in \cite{dean2019sample} for LQR. As we will show, this allows us to quantify the performance degradation due to noise-corrupted data in behavioral models with respect to LQG. The first step is to construct a robust version of \eqref{prob:IOP_data} that is defined in terms of the available noisy historical data. The proof of Proposition~\ref{pr:robust_IOP} is reported in the \preprintswitch{Appendix of \cite{furieri2021behavioral}}{Appendix}. For simplicity, but without loss of generality, we assume that $\mathbf{L}, \mathbf{R}, \bm{\Sigma}_w,\bm{\Sigma}_v$ are  identity matrices with appropriate dimensions.

\begin{proposition}
\label{pr:robust_IOP}
Assume that historical and recent data are affected by noise. Let $\hatbf{G}, \hatbf{y}_{free}$ be estimators of $\mathbf{G},\mathbf{y}_{free}$, respectively, such that Assumption~\ref{ass:bounded_error} holds with $\epsilon>0$. Let $J(\mathbf{G},\mathbf{K}) = \sqrt{\mathbb{E}_{\mathbf{w},\mathbf{v}}\left[\mathbf{y}^\mathsf{T}\mathbf{y}+\mathbf{u}^\mathsf{T}\mathbf{u}\right]}$ denote the square root of the cost in \eqref{eq:cost_output}. Consider the following model-based worst-case robust optimal control problem:
\begin{alignat}{3}
    &\min_{\mathbf{K}}~\max_{\normtwo{\bm{\Delta}} \leq \epsilon,~\normtwo{\bm{\delta}_0} \leq \epsilon}
    && J(\mathbf{G},\mathbf{K})\label{prob:robust_deltas}\\
   &\st~&& 	\eqref{eq:state_compact}, \eqref{eq:output_compact}, \eqref{eq:control_policy}\,. \nonumber
\end{alignat}
Then, problem~\eqref{prob:robust_deltas} is equivalent to 
\begin{equation} 
    \begin{aligned}
        \min_{\hatbm{\Phi}} \; \max_{\normtwo{\bm{\Delta}}\leq \epsilon,~\normtwo{\bm{\delta}_0} \leq \epsilon} \quad &      J(\mathbf{G},\mathbf{K}) = \left\|\begin{bmatrix}\hatbm{\Phi}_{yy}(I-\mathbf{\Delta}\hatbm{\Phi}_{uy})^{-1}&\hatbm{\Phi}_{yy}(I-\mathbf{\Delta}\hatbm{\Phi}_{uy})^{-1}(\hatbf{G}+\mathbf{\Delta})\\\hatbm{\Phi}_{uy}(I-\mathbf{\Delta}\hatbm{\Phi}_{uy})^{-1}&(I-\hatbm{\Phi}_{uy}\mathbf{\Delta})^{-1}\hatbm{\Phi}_{uu}\end{bmatrix}\begin{bmatrix}I&0&\hatbf{y}_{free}+\bm{\delta}_0\\0&I&0\end{bmatrix}\right\|_F   \label{prob:robust_IOP}\\
        \st \quad &  
        \begin{bmatrix}
        	 I&-\hatbf{G}
        	 \end{bmatrix}\begin{bmatrix}
        	 \hatbm{\Phi}_{yy} & \hatbm{\Phi}_{yu}\\\hatbm{\Phi}_{uy} & \hatbm{\Phi}_{uu}
        	 \end{bmatrix}=\begin{bmatrix}
        	 I&0
        	 \end{bmatrix}, \\
        &	 \begin{bmatrix}
        	  \hatbm{\Phi}_{yy} & \hatbm{\Phi}_{yu}\\\hatbm{\Phi}_{uy} & \hatbm{\Phi}_{uu}
        	 \end{bmatrix}
        	 \begin{bmatrix}
        	 -\hatbf{G}\\I
        	 \end{bmatrix}=\begin{bmatrix}
        	 0\\I
        	 \end{bmatrix},  \\
        &	 \hatbm{\Phi}_{yy}, \hatbm{\Phi}_{yu}, \hatbm{\Phi}_{uy}, \hatbm{\Phi}_{uu} \text{\emph{ with causal sparsities.}}   
    \end{aligned}
\end{equation}
\end{proposition}
The robust optimization problem in Proposition~\ref{pr:robust_IOP} is highly non-convex. We therefore proceed with deriving a quasi-convex upperbound to $J(\mathbf{G},\mathbf{K})$ to be used for controller synthesis and suboptimality analysis.
\subsection{A tractable robust BIOP formulation}
The following lemma serves as the basis to derive a tractable formulation of \eqref{prob:robust_IOP}. Its rather lengthy technical proof is reported in the \preprintswitch{Appendix of \cite{furieri2021behavioral}}{Appendix}.
\begin{lemma}
\label{le:upperbound}
 Let $\epsilon = \max(\epsilon_G,\epsilon_0)$ and assume $\epsilon \norm{\hatbm{\Phi}_{uy}}_2< 1$. Further assume that $\norm{\hatbm{\Phi}_{uy}}_2\leq \alpha$ for $\alpha >0$. Then, we have
 \begin{equation}
     \label{eq:nonconvex_bound}
     J(\mathbf{G},\mathbf{K}) \leq  \frac{1}{1-\epsilon \norm{\hatbm{\Phi}_{uy}}_2} \norm{\begin{bmatrix}
        	  \sqrt{1+h(\epsilon,\alpha,\hatbf{G}) + h(\epsilon,\alpha,\hatbf{y}_{free})}\hatbm{\Phi}_{yy} & \hatbm{\Phi}_{yu} & \hatbm{\Phi}_{yy}\hatbf{y}_{free}\\\sqrt{1+ h(\epsilon,\alpha,\hatbf{y}_{free})}\hatbm{\Phi}_{uy} & \hatbm{\Phi}_{uu} & \hatbm{\Phi}_{uy}\hatbf{y}_{free}
        	 \end{bmatrix}}_F
 \end{equation}
 where
 \begin{equation*}
     h(\epsilon, \alpha,\mathbf{Y}) = \epsilon^2(2 + \alpha\|\mathbf{Y}\|_2 )^2+2\epsilon\norm{\mathbf{Y}}_{2}(2+\alpha\norm{\mathbf{Y}}_2)\,.
 \end{equation*}
\end{lemma}

Exploiting the reformulation idea first introduced in \cite{matni2017scalable} and utilized for analysis in \cite{zhengfurieri2020sample}, we  are now ready to establish a quasi-convex reformulation of problem~\eqref{prob:robust_IOP}. 

\begin{theorem}
\label{thm: robust formulation}
Given estimation errors $\epsilon_G,\epsilon_0$ with $\epsilon = \max(\epsilon_G,\epsilon_0)$, and for any $\alpha >0$, the minimal cost of problem \eqref{prob:robust_deltas} is upper bounded by the  minimal cost of the following quasi-convex program: 
\begin{alignat}{3}
&\min_{\gamma \in [0,\epsilon^{-1})} \frac{1}{1-\epsilon \gamma} && \min_{\hatbm{\Phi}} \qquad J_{inner} \label{prob:quasi_convex}\\
&\st~&&\begin{bmatrix}
        	 I&-\hatbf{G}
        	 \end{bmatrix}\begin{bmatrix}
        	 \hatbm{\Phi}_{yy} & \hatbm{\Phi}_{yu}\\\hatbm{\Phi}_{uy} & \hatbm{\Phi}_{uu}
        	 \end{bmatrix}=\begin{bmatrix}
        	 I&0
        	 \end{bmatrix}, \nonumber\\
        &~~&&	 \begin{bmatrix}
        	  \hatbm{\Phi}_{yy} & \hatbm{\Phi}_{yu}\\\hatbm{\Phi}_{uy} & \hatbm{\Phi}_{uu}
        	 \end{bmatrix}
        	 \begin{bmatrix}
        	 -\hatbf{G}\\I
        	 \end{bmatrix}=\begin{bmatrix}
        	 0\\I
        	 \end{bmatrix},  \nonumber\\
        &~~&&	 \hatbm{\Phi}_{yy}, \hatbm{\Phi}_{yu}, \hatbm{\Phi}_{uy}, \hatbm{\Phi}_{uu} \text{\emph{ with causal sparsities,}} \nonumber \\
        &~~&& \norm{\hatbm{\Phi}_{uy}}_2 \leq \min(\gamma,\alpha)\,.\nonumber
\end{alignat}
where $J_{inner}$ is equal to
\begin{equation*}
    \norm{\begin{bmatrix}
        	  \sqrt{1\hspace{-0.08cm} +\hspace{-0.08cm} h(\epsilon,\alpha,\hatbf{G})\hspace{-0.08cm} +\hspace{-0.08cm}  h(\epsilon,\alpha,\hatbf{y}_{free})}\hatbm{\Phi}_{yy} & \hatbm{\Phi}_{yu} & \hatbm{\Phi}_{yy}\hatbf{y}_{free}\\\sqrt{1+ h(\epsilon,\alpha,\hatbf{y}_{free})}\hatbm{\Phi}_{uy} & \hatbm{\Phi}_{uu} & \hatbm{\Phi}_{uy}\hatbf{y}_{free}
        	 \end{bmatrix}}_F\hspace{-0.15cm} .
\end{equation*}
\end{theorem}
\begin{proof}
Directly follows from Lemma~\ref{le:upperbound} and \cite[Theorem~3.2]{zhengfurieri2020sample}.
\end{proof}

First, notice that the inner minimization problem in \eqref{prob:quasi_convex} is strongly convex for a fixed $\gamma$, and that the outer function $(1-\epsilon \gamma)^{-1}$ is monotonically increasing in $\gamma$. Hence, it is well-known that the overall program can be efficiently solved by golden search on $\gamma$ and solving the corresponding instances of the inner program. Second,  we explicitly take into account the effect of an unknown and noisy initial state $x(0) \in \mathbb{R}^n$ through the parameter $\hatbf{y}_{free}$. Assuming $x(0)=0$ as per \cite{xue2020data} may not be realistic for practical purposes, as the user initially lets the system free to evolve in order to harvest data. Furthermore, the following analysis will show that, for finite-horizon control problems, the suboptimality strongly depends on $x(0)\in \mathbb{R}^n$ as a function of $\norm{\mathbf{y}_{free}}^2_2$.  Last, we note that the constraint on $\|\hatbm{\Phi}_{uy}\|_2$ is the main source of suboptimality with respect to the true LQG problem~\eqref{prob:IOP}; as pointed out in \cite{dean2019sample,xue2020data,zhengfurieri2020sample}, this additional constraint enforces stronger disturbance rejection properties, for which we have to pay in terms of performance. We are now ready to quantify the suboptimality of \eqref{prob:quasi_convex} with respect to \eqref{prob:IOP}.

\section{Suboptimality Analysis}
\label{sec:suboptimality}
In this section, we denote as $\starbf{K},\starbm{\Phi}$  the optimal controller and corresponding closed-loop responses for the real LQG problem \eqref{prob:IOP}.  Furthermore, we denote as $\hatbf{K}^\star,\hatbm{\Phi}^\star$ the optimal controller and corresponding closed-loop responses for the quasi-convex program \eqref{prob:quasi_convex} and let $J^\star = J(\mathbf{G},\mathbf{K}^\star)$ and $\hat{J} = J(\mathbf{G},\hatbf{K}^\star)$.

Next, inspired by the analysis in \cite{zhengfurieri2020sample}, we show that if $\epsilon$ is small enough it holds
\begin{equation*}
    \frac{\hat{J}^2-{J^\star}^2}{{J^\star}^2} = \mathcal{O}\left(\epsilon\right)\,.
\end{equation*}
In other words, for a small estimation error $\epsilon$ on the impulse response, applying controller $\hatbf{K}^\star$ (which is solely computed with noisy data) to the \emph{real} plant achieves almost optimal closed-loop performance.

We start with a lemma that analytically characterizes a feasible solution to problem \eqref{prob:quasi_convex}. We then proceed with characterizing the suboptimality bound. The proofs of Lemma~\ref{le:feasible} and Theorem~\ref{th:suboptimality} are reported in the \preprintswitch{Appendix of \cite{furieri2021behavioral}}{Appendix}.
\begin{lemma}[Feasible solution]
\label{le:feasible}
Let $\eta = \epsilon \normtwo{\starbm{\Phi}_{uy}}$, and select $\alpha \geq \sqrt{2}\frac{\eta}{\epsilon(1-\eta)}$. Then, if $\eta< \frac{1}{5}$, the following expressions
    \begin{align}
        \tildebm{\Phi}_{yy} &= \starbm{\Phi}_{yy} (I \hspace{-0.08cm}+ \hspace{-0.08cm}\mathbf{\Delta} \starbm{\Phi}_{uy})^{-1},~ 
        \tildebm{\Phi}_{yu} = \starbm{\Phi}_{yy} (I \hspace{-0.08cm}+\hspace{-0.08cm} \mathbf{\Delta} \starbm{\Phi}_{uy})^{-1}(\mathbf{G} \hspace{-0.08cm}-\hspace{-0.08cm} \mathbf{\Delta}), \nonumber\\
        \tildebm{\Phi}_{uy} &= \starbm{\Phi}_{uy} (I + \mathbf{\Delta} \starbm{\Phi}_{uy})^{-1},~
        \tildebm{\Phi}_{uu} = (I +  \starbm{\Phi}_{uy}\mathbf{\Delta})^{-1}\starbm{\Phi}_{uu} , \nonumber \\
        \widetilde{\gamma} &= \frac{\sqrt{2}\eta}{\epsilon(1-\eta)},   \label{eq:suboptimal}
    \end{align}
provide a feasible solution to problem \eqref{prob:quasi_convex}.
\end{lemma}

\begin{theorem}
\label{th:suboptimality}
Suppose that $ \frac{5\sqrt{2}}{4}\normtwo{\starbm{\Phi}_{uy}}\leq \alpha \leq 5\normtwo{\starbm{\Phi}_{uy}}$ and that $\epsilon < \frac{1}{5 \normtwo{\starbm{\Phi}_{uy}}}$.
Then, when applying the optimal solution $\hatbf{K}^\star$ of \eqref{prob:quasi_convex} to the true plant $\mathbf{G}$, the relative error with respect to the true optimal cost is upper bounded as
\begin{align*}
        \frac{\hat{J}^2-{J^\star}^2}{{J^\star}^2} &\leq  20\epsilon \normtwo{\starbf{\Phi}_{uy}}  + 4(M+V)\\%20\epsilon_B \normtwo{\starbm{\Phi}_{uy}}+2h(\epsilon,\normtwo{\starbm{\Phi}_{uy}},\normtwo{\starbfcal{O}})+\mathcal{O}(\epsilon^2,\normtwo{\starbm{\Phi}_{uy}}^2,\normtwo{\starbfcal{O}})\\
       % &= \mathcal{O}\left(\epsilon,\normtwo{\starbm{\Phi}_{uy}}^2,\normtwo{\starbfcal{O}}^2\right)\,.
       &= \mathcal{O}\left(\epsilon \norm{\starbm{\Phi}_{uy}}_2(\norm{\mathbf{G}}^2_2+\norm{\mathbf{y}_{free}}^2_2)\right)\,,
\end{align*}
where 
\begin{align*}
   & M = h(\epsilon,\alpha,\hatbf{G})+ h(\epsilon,\alpha,\hatbf{y}_{free})+ h(\epsilon,\norm{\starbm{\Phi}_{uy}}_2,\mathbf{G}) + h(\epsilon,\norm{\starbm{\Phi}_{uy}}_2,\mathbf{y}_{free})\,,\\
      & V = h(\epsilon,\alpha,\hatbf{y}_{free})+h(\epsilon,\norm{\starbm{\Phi}_{uy}}_2,\mathbf{y}_{free})\,,
\end{align*}
and $h(a, b,\mathbf{Y}) = a^2(2 + b\|\mathbf{Y}\|_2 )^2+2a\norm{\mathbf{Y}}_2(2+b\norm{\mathbf{Y}}_2)$.

\end{theorem}

Theorem~\ref{th:suboptimality} shows that the relative performance of the robust BIOP formulation \eqref{prob:quasi_convex} with respect to its exact non-noisy version \eqref{prob:IOP_data}  decreases linearly with $\epsilon$, as long as $\epsilon$ is small enough to guarantee $\epsilon \norm{\bm{\Phi}_{uy}^\star}_2 < \frac{1}{5}$. The bound also grows quadratically with the norm of the true impulse and free responses, which implies that an unstable system will be difficult to control for a long horizon. Note that it is appropriate to choose $\alpha$ not too large, and specifically $\alpha\leq 5\norm{\bm{\Phi}_{uy}}_2<\epsilon^{-1}$ in order for the scaling of $h(\epsilon,\alpha,\hatbf{G})$ in terms of $\epsilon$ not to dominate over $h(\epsilon,\norm{\starbm{\Phi}_{uy}}_2,\mathbf{G})$. Our rate in terms of $\epsilon$ matches that of \cite{zhengfurieri2020sample,dean2019sample}, which are valid in infinite-horizon. In spite of the additional challenges of considering a noisy unknown initial state $x(0) \in \mathbb{R}^n$ and noisy output-feedback, our rate also matches the one achieved with the approach of \cite{xue2020data} valid for $x(0) = 0$ and state-feedback.  

\begin{remark}[Sample complexity]
In  related work, e.g. \cite{dean2019sample,zhengfurieri2020sample,xue2020data}, the authors more precisely quantify $\epsilon$ and the probability of the estimate to be within the corresponding norm error interval as a function of the noise statistics and the real system parameters, leading to an end-to-end sample complexity analysis. This is achieved by focusing on a specific estimation technique (i.e. least squares in \cite{dean2019sample,zhengfurieri2020sample} and column averaging in \cite{xue2020data}) and the corresponding non-asymptotic norm error bounds \cite{oymak2019non,tropp2012user}. We expect that analogous results can be derived for the least-square choice $(G,g) = (G_{LS},g_{LS})$. However, in this work we wished to focus on the potential generality of the proposed BIOP, i.e., the fact that the approximation of the impulse and free responses is not bound to a specific estimation technique. Hence, here we have limited ourselves to  deriving a suboptimality bound as a function of $\epsilon$, and will not further characterize $\epsilon$ and the success probability, as both are dependent on the chosen estimation technique.
\end{remark}

  \begin{figure*}[th!]
     \centering
     \includegraphics[width = 0.75\textwidth]{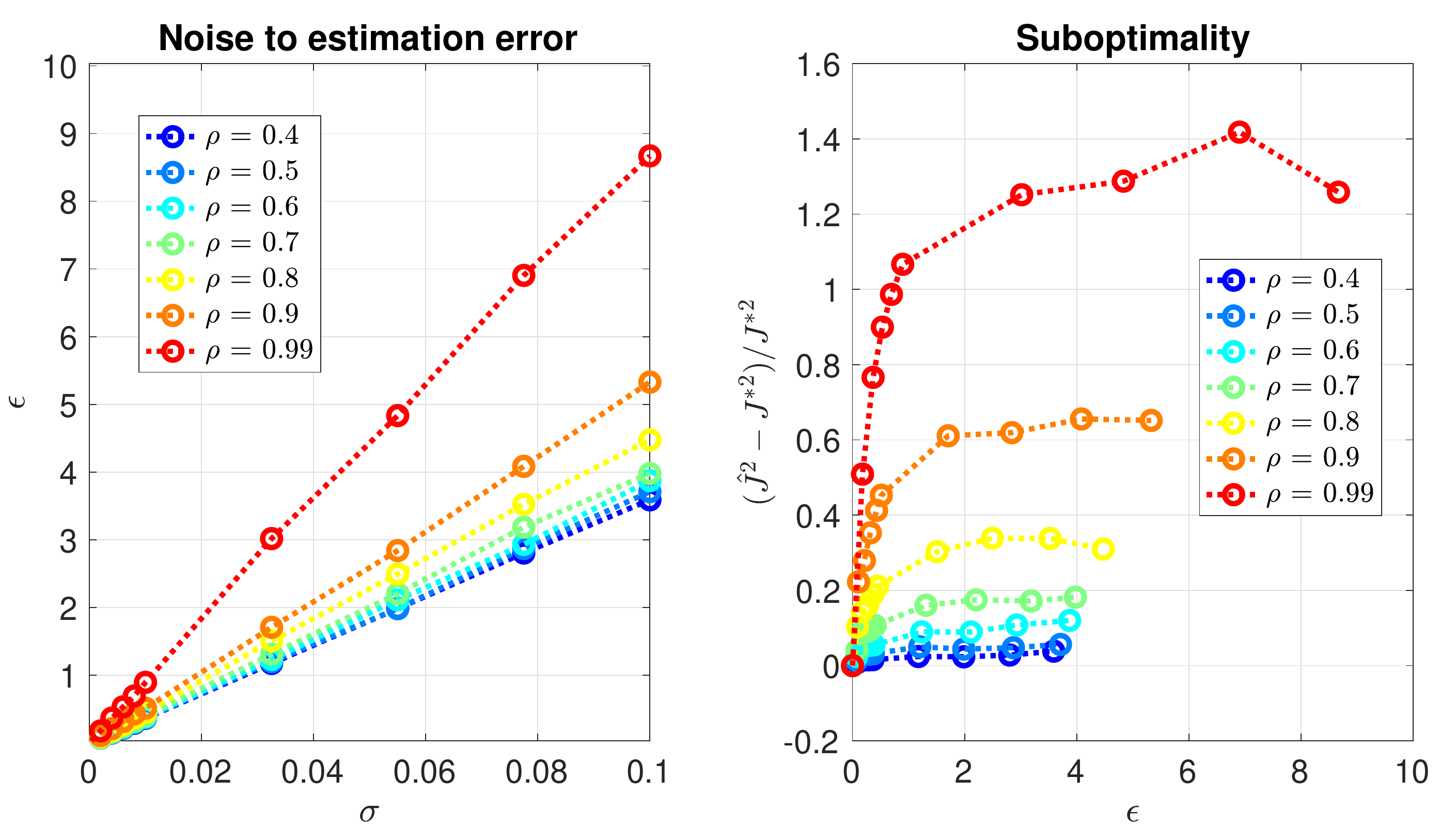}
     
     \caption{ Estimation error as a function of the noise level (on the left). Corresponding suboptimality gap for increasing values of the spectral radius $\rho$ of matrix $A$ (on the right).}
     \label{fig:RobustBIOPPerformance}
 \end{figure*}
\section{Numerical Experiments}
\label{sec:experiments}

In this section we present numerical results validating our theoretical analysis. For solving optimization problems we used MOSEK \cite{mosek}, called through MATLAB via YALMIP \cite{YALMIP} on a standard laptop computer\footnote{The code is open-source and available at \url{https://gitlab.nccr-automation.ch/data-driven-control-epfl/biop}}. Our goals are  1) to verify the noiseless BIOP formulation in Theorem~\ref{th:BIOP} and 2) to validate the suboptimality analysis of Theorem~\ref{th:suboptimality} in the presence of noise-corrupted data. In the experiments, we considered the LTI system characterized by the matrices
\begin{align*}
&A = \rho \begin{bmatrix}
       0.8 &  0.4\\0.8 & -0.6
\end{bmatrix}, \text{ } B =  \begin{bmatrix}
      1 &  0.2\\ 
      2 &  0.3
\end{bmatrix},\text{ } C =  \begin{bmatrix}
       1 & 1\\0.7 &  0.2
\end{bmatrix}\,.
\end{align*}
It can be verified that the value $\rho>0$ corresponds to the spectral radius of $A$. The cost function is given by \eqref{eq:cost_output}, where  $N = 11$ and the cost weights are chosen as $L_t = I_p$  and $R_t=I_m$  for every $t = 0,\ldots,10$. The average in \eqref{eq:cost_output} is taken over future input/output noise with  variances $\Sigma_w = I_m$ and $\Sigma_v = I_m $. Assuming an initial state $x(0) = \begin{bmatrix}1&-1\end{bmatrix}^\mathsf{T}$ and $\rho = 0.99$, the optimal controller $\mathbf{K}^\star$ can be found by solving the model-based optimization problem \eqref{prob:IOP}, and the corresponding optimal cost is $J^\star = 12.8006$. 

Hereafter, we assume that the system parameters $A$, $B$, $C$ and $x(0)$ are completely \emph{unknown}. Instead, the following data are available: 1) a \emph{historical} system trajectory $        \{\mathbf{y}^h_{[0,T-1]},\mathbf{u}^h_{[0,T-1]}\}$, with $\mathbf{y}^h_{[0,T-1]} = \mathbf{y}_{[-T_{h},-T_{h}+T-1]}$ and $\mathbf{u}^h_{[0,T-1]} = \mathbf{u}_{[-T_{h},-T_h+T-1]}$ where $T = 200$ and $T_h = 249$, and 2) a \emph{recent} system trajectory $\{\mathbf{y}^r_{[0,T_{ini}-1]},\mathbf{u}^r_{[0,T_{ini}-1]}\}$, with $\mathbf{y}^r_{[0,T_{ini}-1]} = \mathbf{y}_{[-T_{ini},-1]}$, $\mathbf{u}^r_{[0,T_{ini}-1]} = \mathbf{u}_{[-T_{ini},-1]}$ and $T_{ini} = 30$. When the collected data are \emph{noiseless}, one can compute a solution $(G,g)$ to \eqref{eq:BM}, for instance by using \eqref{eq: LS_for_G}, and solve the optimization problem \eqref{prob:IOP_data} to find the optimal closed-loop responses. In this case, the solution $(\starbm{\Phi}_{yy}, \starbm{\Phi}_{yu}, \starbm{\Phi}_{uy}, \starbm{\Phi}_{uu})$  yields the optimal closed-loop control policy $\mathbf{K}^\star = \starbm{\Phi}_{uy}(\bm{\Phi}^\star_{yy})^{-1}$ and the same optimal cost $J^\star = 12.8006 $ obtained before, as predicted by Theorem~\ref{th:BIOP}.

 We now focus on the case where the historical and recent data are affected by noise. The corrupting noise has zero expected value, and variances equal to $\bm{\Sigma}^h_w=\bm{\Sigma}^r_w=\sigma I_m$, $\bm{\Sigma}^h_v = \bm{\Sigma}^r_v = \sigma I_p$. We analyze performance degradation for increasing values of $\sigma$. First, we note that solving \eqref{prob:IOP_data} with noisy data yields unsatisfactory results; indeed, the problem is often infeasible due to an incoherent estimation of $\mathbf{G}$. Next, we consider the robust formulation of Theorem~\ref{thm: robust formulation}. 
 
We estimate an error level $\epsilon$ that is valid with high probability by computing $$\tilde{\epsilon}=\max \left( \norm{\text{Toep}(\widehat{Y}_fG_{LS})-\mathbf{G}}_2,\norm{\widehat{Y}_fg_{LS}-\mathbf{y}_{free}}_2\right)\,,$$ many times over different realizations of the corrupting noise, and selecting $\epsilon$ as a high percentile value, e.g. the $90$-th percentile, of the $\tilde{\epsilon}$'s.\footnote{ It would be more realistic to implement a proper bootstrap procedure, see \cite{efron1992bootstrap} or the approach in \cite{dean2019sample}.  Since our focus is to validate the theoretical bounds, we leave implementing these methods as future work. } While computing the least square solution $(G_{LS},g_{LS})$ is a viable choice, we highlight that our method is compatible with the more refined approaches of \cite{yin2020maximum,alpago2020extended,coulson2020distributionally} to obtain lower values for $\epsilon$.

 The hyper-parameter $\alpha$ can be tuned manually until satisfactory results are obtained and $\alpha< \epsilon^{-1}$ is verified. In Figure~\ref{fig:RobustBIOPPerformance}, we report the suboptimality gap one incurs by applying the controller $\hatbf{K}^\star$ that solves the robust BIOP \eqref{prob:robust_IOP}. Specifically, for each choice of the spectral radius $\rho= 0.4,0.5,\ldots,0.9,0.99$, we consider increasing levels of the variance $\sigma^2$ of the noise that corrupts the historical and recent data. We first plot the corresponding estimation errors $\epsilon$ on the left part of Figure~\ref{fig:RobustBIOPPerformance}. While observing that $\epsilon$ grows almost linearly with $\sigma$ for any fixed $\rho$, we highlight that a formal analysis of this relationship is beyond the scope of this paper. We then plot the suboptimality gap $\frac{\hat{J}^2-{J^\star}^2}{{J^\star}^2}$ as a function of $\epsilon$ on the right part of Figure~\ref{fig:RobustBIOPPerformance}. It can be observed that, as predicted by Theorem~\ref{th:suboptimality}, 1) the gap linearly converges to $0$ as $\epsilon$ converges to $0$, and 2) for similar values of  $\epsilon$, the gap grows faster than linearly with the spectral radius $\rho$. We finally observe that, in theory, the BIOP and robust BIOP formulations in finite-horizon are valid for unstable systems with $\rho>1$. However, in practice, it is inherently challenging to collect trajectories  of an unstable system, as the values to be plugged into the corresponding numerical programs will become too large to be handled by numerical solvers.  For unstable systems in a data-driven scenario, it is common to assume knowledge of a pre-stabilizing controller \cite{simchowitz2020improper,zhengfurieri2020sample}.

\section{Conclusions}
\label{sec:conclusions}
We have proposed the BIOP, a method for the  design of optimal  output-feedback controllers which directly embeds historical input-output trajectories in its formulation. When these historical data are noiseless, the BIOP is equivalent to the standard IOP and recovers an optimal LQG controller. In the presence of noise-corrupted data, we propose a robust version of the BIOP that explicitly incorporates the estimated uncertainty level and that can be solved efficiently through convex programming. By exploiting recently developed analysis techniques, the suboptimality of the obtained solution is quantified and compared with the nominal LQG solution. Furthermore, the developed framework is readily compatible with state-of-the-art behavioral estimation and prediction techniques, e.g. \cite{coulson2020distributionally,alpago2020extended,yin2020maximum}.

\section*{Acknowledgments}
We thank Mingzhou Yin, Andrea Iannelli and Roy Smith for helpful discussion.

\bibliographystyle{IEEEtran}
\bibliography{references}

 \newpage

\preprintswitch{}{
\appendix

\subsection{Proof of Proposition~\ref{prop:IOP}}
\label{app:prop:IOP}
    For the first statement, notice that the controller $\mathbf{K}$ achieves the closed-loop responses \eqref{eq:CL_responses}. Now select $(\bm{\Phi}_{yy}, \bm{\Phi}_{yu}, \bm{\Phi}_{uy}, \bm{\Phi}_{uu})$ as
    \begin{equation}
    \begin{bmatrix}
           \bm{\Phi}_{yy} & \bm{\Phi}_{yu} \\
            \bm{\Phi}_{uy} & \bm{\Phi}_{uu} 
        \end{bmatrix}=\begin{bmatrix}(I-\mathbf{GK})^{-1} & (I-\mathbf{GK})^{-1}\mathbf{G}\\ \mathbf{K}(I-\mathbf{GK})^{-1} & (I-\mathbf{KG})^{-1}\end{bmatrix}\,.
        \label{eq:CL_responses_proofs}
    \end{equation}
   Clearly, $\mathbf{K} = \bm{\Phi}_{uy}\bm{\Phi}_{yy}^{-1}$, and by plugging the corresponding expressions in \eqref{eq:ach1}-\eqref{eq:ach3}, we verify that \eqref{eq:ach1}-\eqref{eq:ach3} are satisfied.

   For the second statement, it is easy to notice   $\mathbf{K}$ is causal by construction because $\bm{\Phi}_{uy}$ and $\bm{\Phi}_{yy}$ are block lower-triangular. % One can then verify, for example:
   Consider now the equation $\bm{\Phi}_{yy} = (I-\mathbf{GK})^{-1}$ corresponding to the upper-left block of \eqref{eq:CL_responses_proofs}. By selecting the controller $\mathbf{K} = \bm{\Phi}_{uy}\bm{\Phi}_{yy}^{-1}$ one has  \begin{align*}
    (I-\mathbf{G}\bm{\Phi}_{uy}\bm{\Phi}_{yy}^{-1})^{-1} &= (I-\mathbf{G}\bm{\Phi}_{uy}(I+\mathbf{G}\bm{\Phi}_{uy})^{-1})^{-1}\\
    &= ((I+\mathbf{G}\bm{\Phi}_{uy}-\mathbf{G}\bm{\Phi}_{uy})(I+\mathbf{G}\bm{\Phi}_{uy})^{-1})^{-1}\\
    &= I +\mathbf{G}\bm{\Phi}_{uy} = \bm{\Phi}_{yy}\,,
    \end{align*}
    which shows that $\bm{\Phi}_{yy}$ is the closed-loop response from $\mathbf{v}_{[0,N-1]}+\mathbf{CP}_A(:,0)x(0)$ to $\mathbf{y}_{[0,N-1]}$ as per \eqref{eq:CL_responses}. Similar computations for the remaining closed-loop responses conclude the proof.

    \subsection{Proof of Proposition~\ref{prop:strongly_convex}}
Let $\bm{\delta}_y = \mathbf{v}_{[0,N-1]}+\mathbf{CP}_A(:,0)x(0)$ and $\bm{\delta}_u = \mathbf{w}_{[0,N-1]}$. From linearity of the expectation operator it follows that
\begin{align}
&J(\bm{\Phi}_{yy}, \bm{\Phi}_{yu}, \bm{\Phi}_{uy}, \bm{\Phi}_{uu})=\mathbb{E}_{\bm{\delta}_y,\bm{\delta}_u}[\mathbf{y}^\mathsf{T}\mathbf{L}\mathbf{y}+\mathbf{u}^\mathsf{T}\mathbf{R}\mathbf{u}] \nonumber\\
&= \mathbb{E}[\mathbf{L}^{\frac{1}{2}}\bm{\delta}_y^\mathsf{T} \bm{\Phi}_{yy}^\mathsf{T}\bm{\Phi}_{yy}\bm{\delta}_y \mathbf{L}^{\frac{1}{2}}]+\mathbb{E}[\mathbf{L}^{\frac{1}{2}}\bm{\delta}_u^\mathsf{T} \bm{\Phi}_{yu}^\mathsf{T}\bm{\Phi}_{yu}\bm{\delta}_u \mathbf{L}^{\frac{1}{2}}]+\mathbb{E}[\mathbf{R}^{\frac{1}{2}}\bm{\delta}_y^\mathsf{T} \bm{\Phi}_{uy}^\mathsf{T}\bm{\Phi}_{uy}\bm{\delta}_y \mathbf{R}^{\frac{1}{2}}]+\mathbb{E}[\mathbf{R}^{\frac{1}{2}}\bm{\delta}_u^\mathsf{T} \bm{\Phi}_{uu}^\mathsf{T}\bm{\Phi}_{uu}\bm{\delta}_u \mathbf{R}^{\frac{1}{2}}]\nonumber\\%\,.
&= \mathbf{L}^{\frac{1}{2}}\mathbb{E}[\bm{\delta}_y^\mathsf{T} \bm{\Phi}_{yy}^\mathsf{T}\bm{\Phi}_{yy}\bm{\delta}_y]\mathbf{L}^{\frac{1}{2}} + \mathbf{L}^{\frac{1}{2}}\mathbb{E}[\bm{\delta}_u^\mathsf{T} \bm{\Phi}_{yu}^\mathsf{T}\bm{\Phi}_{yu}\bm{\delta}_u]\mathbf{L}^{\frac{1}{2}} + \mathbf{R}^{\frac{1}{2}}\mathbb{E}[\bm{\delta}_y^\mathsf{T} \bm{\Phi}_{uy}^\mathsf{T}\bm{\Phi}_{uy}\bm{\delta}_y]\mathbf{R}^{\frac{1}{2}} + \mathbf{R}^{\frac{1}{2}}\mathbb{E}[\bm{\delta}_u^\mathsf{T} \bm{\Phi}_{uu}^\mathsf{T}\bm{\Phi}_{uu}\bm{\delta}_u]\mathbf{R}^{\frac{1}{2}}\,.\label{eq:cost_addends}
% &= \mathbf{L}^{\frac{1}{2}}[\Tr(\bm{\Phi}_{yy}^\mathsf{T}\bm{\Phi}_{yy}\bm{\Sigma}_v) + (\mathbf{CP}_{A}(:,0)x(0))^\mathsf{T} \bm{\Phi}_{yy}^\mathsf{T}\bm{\Phi}_{yy}\mathbf{CP}_{A}(:,0)x(0) ]\mathbf{L}^{\frac{1}{2}} + \dots\\
% &= \mathbf{L}^{\frac{1}{2}}\frobenius{\bm{\Phi}_{yy}\bm{\Sigma}_v^{\frac{1}{2}}}^2\mathbf{L}^{\frac{1}{2}} + \mathbf{L}^{\frac{1}{2}}\normtwo{\bm{\Phi}_{yy}\mathbf{CP}_{A}(:,0)x(0)}^2\mathbf{L}^{\frac{1}{2}} + \dots
\end{align}
Focusing, for example, on the first addend we have
\begin{align*}
    \mathbf{L}^{\frac{1}{2}}\mathbb{E}[\bm{\delta}_y^\mathsf{T} \bm{\Phi}_{yy}^\mathsf{T}\bm{\Phi}_{yy}\bm{\delta}_y]\mathbf{L}^{\frac{1}{2}} &= \mathbf{L}^{\frac{1}{2}}[\Tr(\bm{\Phi}_{yy}^\mathsf{T}\bm{\Phi}_{yy}\bm{\Sigma}_v) + (\mathbf{CP}_{A}(:,0)x(0))^\mathsf{T} \bm{\Phi}_{yy}^\mathsf{T}\bm{\Phi}_{yy}\mathbf{CP}_{A}(:,0)x(0) ]\mathbf{L}^{\frac{1}{2}}\\
    &= \mathbf{L}^{\frac{1}{2}}\frobenius{\bm{\Phi}_{yy}\bm{\Sigma}_v^{\frac{1}{2}}}^2\mathbf{L}^{\frac{1}{2}} + \mathbf{L}^{\frac{1}{2}}\normtwo{\bm{\Phi}_{yy}\mathbf{CP}_{A}(:,0)x(0)}^2\mathbf{L}^{\frac{1}{2}}\\
    &= \mathbf{L}^{\frac{1}{2}}\frobenius{\bm{\Phi}_{yy}\bm{\Sigma}_v^{\frac{1}{2}}}^2\mathbf{L}^{\frac{1}{2}} + \mathbf{L}^{\frac{1}{2}}\frobenius{\bm{\Phi}_{yy}\mathbf{CP}_{A}(:,0)x(0)}^2\mathbf{L}^{\frac{1}{2}}\,,
\end{align*}
where the first equality follows from $\mathbb{E}_{x}(x^\mathsf{T}Mx) = \Tr(M\Sigma_x) + \mu_x^\mathsf{T}M\mu_x$, where $\Sigma_x$ and $\mu_x$ are the variance and expected value of the random variable $x$ respectively, while the third equality uses the fact that for vectors $x \in \mathbb{R}^n$ we have $\norm{x}_2= \norm{x}_F$. Similar computations hold for the remaining terms of \eqref{eq:cost_addends}. In total, since $\bm{\delta}_u$ has zero mean, the cost is made up of six addends. Since they are all convex functions of $(\bm{\Phi}_{yy}, \bm{\Phi}_{yu}, \bm{\Phi}_{uy}, \bm{\Phi}_{uu})$, and $\mathbf{R}^{\frac{1}{2}}\norm{\bm{\Phi}_{uy}}_F^2 \bm{\Sigma}_v^{\frac{1}{2}}$ is strongly convex, then $J(\cdot)$ is strongly convex and admits a unique global optimum. By using the property that 
\begin{equation*}
    \norm{M}_F^2 + \norm{N}_F^2 = \norm{\begin{bmatrix}M&N\end{bmatrix}}_F^2 = \norm{\begin{bmatrix}M\\N\end{bmatrix}}_F^2\,,
\end{equation*}
we can rewrite the six addends of the cost compactly as the squared Frobenius norm of the $2 \times 3$ block-matrix in \eqref{prob:IOP}.

\subsection{Proof of Proposition~\ref{pr:robust_IOP}}
First, we verify by direct inspection that for any $\mathbf{K}$, the parameters \begin{equation*}
 \hatbm{\Phi}=\begin{bmatrix}(I-\hatbf{G}\mathbf{K})^{-1} & (I-\hatbf{G}\mathbf{K})^{-1}\hatbf{G}\\ \mathbf{K}(I-\hatbf{G}\mathbf{K})^{-1} & (I-\mathbf{KG})^{-1}\end{bmatrix}\,.
\end{equation*} 
satisfy the constraints of \eqref{prob:robust_IOP} and are such that $\mathbf{K} = \hatbm{\Phi}_{uy}\hatbm{\Phi}_{yy}^{-1}$. Therefore, every controller $\mathbf{K}$ is parametrized in problem \eqref{prob:robust_IOP}, irrespective of $\hatbf{G}$. 

We know that for any $\mathbf{K}$, the cost $J(\mathbf{G},\mathbf{K})$ is equivalent to 
\begin{equation}
\label{eq:cost_proof}
    \frobenius{\begin{bmatrix}(I-\mathbf{GK})^{-1} & (I-\mathbf{GK})^{-1}\mathbf{G}\\ \mathbf{K}(I-\mathbf{GK})^{-1} & (I-\mathbf{KG})^{-1}\end{bmatrix}\begin{bmatrix}I&0&\mathbf{CP}_{A}(:,0)x(0)\\0&I&0\end{bmatrix}}\,,
\end{equation}
Now,  we notice that $\mathbf{G} = \hatbf{G}+\bm{\Delta}$, $\mathbf{y}_{free}= \hatbf{y}_{free}+\bm{\delta}_0$ and substitute into \eqref{eq:cost_proof}. We obtain:

\begin{equation*}
    \begin{split}
        \bm{\Phi}_{yy} &= (I-\mathbf{GK})^{-1} = \left(I-(\hatbf{G} + \mathbf{\Delta})\hatbm{\Phi}_{uy}\hatbm{\Phi}_{yy}^{-1}\right)^{-1} = \left(I-\hatbf{G} \hatbm{\Phi}_{uy}\hatbm{\Phi}_{yy}^{-1} - \mathbf{\Delta}\hatbm{\Phi}_{uy}\hatbm{\Phi}_{yy}^{-1}\right)^{-1}\\
        &= \left((\underbrace{\hatbm{\Phi}_{yy}-\hatbf{G} \hatbm{\Phi}_{uy}}_{I} - \mathbf{\Delta}\hatbm{\Phi}_{uy})\hatbm{\Phi}_{yy}^{-1}\right)^{-1} = \hatbm{\Phi}_{yy}\left(I - \mathbf{\Delta}\hatbm{\Phi}_{uy}\right)^{-1}\,,\\
        %%%
        \bm{\Phi}_{yu} &= (I-\mathbf{GK})^{-1}\mathbf{G} = \bm{\Phi}_{yy}\mathbf{G} = \hatbm{\Phi}_{yy}\left(I - \mathbf{\Delta}\hatbm{\Phi}_{uy}\right)^{-1}(\hatbf{G} + \mathbf{\Delta})\,,\\
        %%%
        \bm{\Phi}_{uy} &= \mathbf{K}(I-\mathbf{GK})^{-1} = \mathbf{K}\bm{\Phi}_{yy} = \hatbm{\Phi}_{uy}\hatbm{\Phi}_{yy}^{-1}\hatbm{\Phi}_{yy}\left(I - \mathbf{\Delta}\hatbm{\Phi}_{uy}\right)^{-1} = \hatbm{\Phi}_{uy}\left(I - \mathbf{\Delta}\hatbm{\Phi}_{uy}\right)^{-1}\,,\\
        %%%
        \bm{\Phi}_{uu} &= (I-\mathbf{KG})^{-1} = \mathbf{K}(I-\mathbf{GK})^{-1}\mathbf{G} + I = \bm{\Phi}_{uy}\mathbf{G} + I = \hatbm{\Phi}_{uy}\left(I - \mathbf{\Delta}\hatbm{\Phi}_{uy}\right)^{-1}(\hatbf{G} + \mathbf{\Delta}) + I\\
        &= \left(I - \hatbm{\Phi}_{uy}\mathbf{\Delta}\right)^{-1}\hatbm{\Phi}_{uy}(\hatbf{G} + \mathbf{\Delta}) + I = \left(I - \hatbm{\Phi}_{uy}\mathbf{\Delta}\right)^{-1}(\hatbm{\Phi}_{uy}\hatbf{G} + \hatbm{\Phi}_{uy}\mathbf{\Delta} + I - \hatbm{\Phi}_{uy}\mathbf{\Delta})\\
        &= \left(I - \hatbm{\Phi}_{uy}\mathbf{\Delta}\right)^{-1}(\hatbm{\Phi}_{uy}\hatbf{G} + I) = \left(I - \hatbm{\Phi}_{uy}\mathbf{\Delta}\right)^{-1} \hatbm{\Phi}_{uu}\,.
    \end{split}
\end{equation*}
This concludes the proof.

\subsection{Proof of Lemma~\ref{le:upperbound}}
The objective function in Proposition~\ref{pr:robust_IOP} can be written as 
\begin{equation*}
    % \color{blue}
    % \begin{split}
        J(\mathbf{G},\mathbf{K}) %&
        = \left\|\begin{bmatrix} \hatbm{\Phi}_{yy}(I-\mathbf{\Delta}\hatbm{\Phi}_{uy})^{-1} & \hatbm{\Phi}_{yy}(I-\mathbf{\Delta}\hatbm{\Phi}_{uy})^{-1}(\hatbf{G}+\mathbf{\Delta}) &  \hatbm{\Phi}_{yy}(I-\mathbf{\Delta}\hatbm{\Phi}_{uy})^{-1}(\hatbf{y}_{free}+\bm{\delta}_0)\\
    \hatbm{\Phi}_{uy}(I-\mathbf{\Delta}\hatbm{\Phi}_{uy})^{-1} & (I-\hatbm{\Phi}_{uy}\mathbf{\Delta})^{-1}\hatbm{\Phi}_{uu} & \hatbm{\Phi}_{uy}(I-\mathbf{\Delta}\hatbm{\Phi}_{uy})^{-1}(\hatbf{y}_{free}+\bm{\delta}_0)\end{bmatrix}\right\|_F\,,
    % &\leq \frac{1}{1-\epsilon \norm{\hatbm{\Phi}_{uy}}_2} \sqrt{\norm{\begin{bmatrix}
    %     	  \hatbm{\Phi}_{yy} & \hatbm{\Phi}_{yu} & \hatbm{\Phi}_{yy}\hatbf{y}_{free}\\\hatbm{\Phi}_{uy} & \hatbm{\Phi}_{uu} & \hatbm{\Phi}_{uy}\hatbf{y}_{free}
    %     	 \end{bmatrix}}_F^2 + h_1 + h_2 + h_3}\\
    % & h_1 =  2\underbrace{\|\hatbm{\Phi}_{yu}\|_F}_{\textcolor{magenta}{=\norm{\hatbm{\Phi}_{yy}\hatbf{G}}_F\leq \norm{\hatbm{\Phi}_{yy}}_F\norm{\hatbf{G}}_2}}\epsilon \|\hatbm{\Phi}_{yy}\|_F (2 + \|\hatbm{\Phi}_{uy}\|_2\|\hatbf{G}\|_2 ) + \left(\epsilon \|\hatbm{\Phi}_{yy}\|_F (2 + \|\hatbm{\Phi}_{uy}\|_2\|\hatbf{G}\|_2 )\right)^2
    % \end{split}
\end{equation*}
or, equivalently, as the square-root of the sum of the square of the Frobenius norms of each of its six blocks. For the upper-left block, we have
\begin{align*}
    \|\hatbm{\Phi}_{yy}(I-\mathbf{\Delta}\hatbm{\Phi}_{uy})^{-1}\|_F &\leq \|\hatbm{\Phi}_{yy}\|_F \norm{\sum_{k=0}^\infty (\bm{\Delta} \hatbm{\Phi}_{uy})^k}_2 \leq \|\hatbm{\Phi}_{yy}\|_F \sum_{k=0}^\infty \norm{(\epsilon_G \hatbm{\Phi}_{uy})}^k_2 = \frac{\|\hatbm{\Phi}_{yy}\|_F}{1-\epsilon_G \|\hatbm{\Phi}_{uy}\|_2}\leq \frac{\|\hatbm{\Phi}_{yy}\|_F}{1-\epsilon \|\hatbm{\Phi}_{uy}\|_2}\,,
    % \norm{\hatbm{\Phi}_{yy}(I-\mathbf{\Delta}\hatbm{\Phi}_{uy})^{-1}}_F &\leq \norm{\hatbm{\Phi}_{yy}}_F \norm{\sum_{k=0}^\infty (\bm{\Delta} \hatbm{\Phi}_{uy})^k}_2 \leq \norm{\hatbm{\Phi}_{yy}}_F \sum_{k=0}^\infty \norm{(\epsilon_G \hatbm{\Phi}_{uy})}^k_2 = \frac{\norm{\hatbm{\Phi}_{yy}}_F}{1-\epsilon_G \norm{\hatbm{\Phi}_{uy}}_2}\leq \frac{\norm{\hatbm{\Phi}_{yy}}_F}{1-\epsilon \norm{\hatbm{\Phi}_{uy}}_2}\,,
\end{align*}
where the convergence of the Neumann series follows from $\bm{\Delta}$ and $\hatbm{\Phi}_{uy}$ having zero-entries diagonal blocks by construction. Similarly
\begin{align*}
    &\|\hatbm{\Phi}_{uy}(I-\mathbf{\Delta}\hatbm{\Phi}_{uy})^{-1}\|_F \leq \frac{\|\hatbm{\Phi}_{uy}\|_F}{1-\epsilon \|\hatbm{\Phi}_{uy}\|_2}\,,\\
    &\|(I-\hatbm{\Phi}_{uy}\mathbf{\Delta})^{-1}\hatbm{\Phi}_{uu}\|_F \leq \frac{\|\hatbm{\Phi}_{uu}\|_F}{1-\epsilon \|\hatbm{\Phi}_{uy}\|_2}\,.
    % &\norm{\hatbm{\Phi}_{uy}(I-\mathbf{\Delta}\hatbm{\Phi}_{uy})^{-1}}_F \leq \frac{\norm{\hatbm{\Phi}_{uy}}_F}{1-\epsilon \norm{\hatbm{\Phi}_{uy}}_2}\,,\\
    % &\norm{(I-\hatbm{\Phi}_{uy}\mathbf{\Delta})^{-1}\hatbm{\Phi}_{uu} }_F \leq \frac{\norm{\hatbm{\Phi}_{uu}}_F}{1-\epsilon \norm{\hatbm{\Phi}_{uy}}_2}\,.
\end{align*}
Next, we have
\begin{align*}
       \|\hatbm{\Phi}_{yy}(I-\mathbf{\Delta}\hatbm{\Phi}_{uy})^{-1}(\hatbf{G}+\mathbf{\Delta})\|_F
       &\leq \|\hatbm{\Phi}_{yy}\hatbf{G}\|_F +  \|\hatbm{\Phi}_{yy}\mathbf{\Delta}\|_F + \left\|\hatbm{\Phi}_{yy}\left(\sum_{k=1}^\infty(\mathbf{\Delta}\hatbm{\Phi}_{uy})^k\right) (\hatbf{G}+\mathbf{\Delta})\right\|_F \\
       &\leq \|\hatbm{\Phi}_{yu}\|_F +  \epsilon \|\hatbm{\Phi}_{yy}\|_F + \|\hatbm{\Phi}_{yy}\|_F\left(\sum_{k=1}^\infty \epsilon^k \|\hatbm{\Phi}_{uy}\|_2^k\right) (\|\hatbf{G}\|_2+\epsilon) \\
       &= \|\hatbm{\Phi}_{yu}\|_F + \epsilon \|\hatbm{\Phi}_{yy}\|_F + \|\hatbm{\Phi}_{yy}\|_F \frac{\epsilon \|\hatbm{\Phi}_{uy}\|_2(\|\hatbf{G}\|_2 + \epsilon)}{1 - \epsilon \|\hatbm{\Phi}_{uy}\|_2}\\
       &\leq  \frac{\|\hatbm{\Phi}_{yu}\|_F + \epsilon \|\hatbm{\Phi}_{yy}\|_F + \epsilon \|\hatbm{\Phi}_{yy}\|_F \|\hatbm{\Phi}_{uy}\|_2(\|\hatbf{G}\|_2 + \epsilon)}{1 - \epsilon \|\hatbm{\Phi}_{uy}\|_2} \\
       &= \frac{\|\hatbm{\Phi}_{yu}\|_F + \epsilon \|\hatbm{\Phi}_{yy}\|_F + \epsilon \|\hatbm{\Phi}_{yy}\|_F \|\hatbm{\Phi}_{uy}\|_2\|\hatbf{G}\|_2   + \epsilon^2\|\hatbm{\Phi}_{uy}\|_2 \|\hatbm{\Phi}_{yy}\|_F  }{1 - \epsilon \|\hatbm{\Phi}_{uy}\|_2} \\
       & \leq  \frac{\|\hatbm{\Phi}_{yu}\|_F + \epsilon \|\hatbm{\Phi}_{yy}\|_F (2 + \|\hatbm{\Phi}_{uy}\|_2\|\hatbf{G}\|_2 )}{1 - \epsilon \|\hatbm{\Phi}_{uy}\|_2}\,,
    %   &\quad \|\hatbm{\Phi}_{yy}(I-\mathbf{\Delta}\hatbm{\Phi}_{uy})^{-1}(\hatbf{G}+\mathbf{\Delta})\|_F \\
    %   &\leq \|\hatbm{\Phi}_{yy}\hatbf{G}\|_F +  \|\hatbm{\Phi}_{yy}\mathbf{\Delta}\|_F + \left\|\hatbm{\Phi}_{yy}\left(\sum_{k=1}^\infty(\mathbf{\Delta}\hatbm{\Phi}_{uy})^k\right) (\hatbf{G}+\mathbf{\Delta})\right\|_F \\
    %   &\leq \|\hatbm{\Phi}_{yu}\|_F +  \epsilon \|\hatbm{\Phi}_{yy}\|_F + \|\hatbm{\Phi}_{yy}\|_F\left(\sum_{k=1}^\infty \epsilon^k \|\hatbm{\Phi}_{uy}\|_2^k\right) (\|\hatbf{G}\|_2+\epsilon) \\
    %   &= \|\hatbm{\Phi}_{yu}\|_F + \epsilon \|\hatbm{\Phi}_{yy}\|_F + \|\hatbm{\Phi}_{yy}\|_F \frac{\epsilon \|\hatbm{\Phi}_{uy}\|_2(\|\hatbf{G}\|_2 + \epsilon)}{1 - \epsilon \|\hatbm{\Phi}_{uy}\|_2}\\
    %   &\leq  \frac{\|\hatbm{\Phi}_{yu}\|_F + \epsilon \|\hatbm{\Phi}_{yy}\|_F + \epsilon \|\hatbm{\Phi}_{yy}\|_F \|\hatbm{\Phi}_{uy}\|_2(\|\hatbf{G}\|_2 + \epsilon)}{1 - \epsilon \|\hatbm{\Phi}_{uy}\|_2} \\
    %   &= \frac{\|\hatbm{\Phi}_{yu}\|_F + \epsilon \|\hatbm{\Phi}_{yy}\|_F + \epsilon \|\hatbm{\Phi}_{yy}\|_F \|\hatbm{\Phi}_{uy}\|_2\|\hatbf{G}\|_2   + \epsilon^2\|\hatbm{\Phi}_{uy}\|_2 \|\hatbm{\Phi}_{yy}\|_F  }{1 - \epsilon \|\hatbm{\Phi}_{uy}\|_2} \\
    %   & \leq  \frac{\|\hatbm{\Phi}_{yu}\|_F + \epsilon \|\hatbm{\Phi}_{yy}\|_F (2 + \|\hatbm{\Phi}_{uy}\|_2\|\hatbf{G}\|_2 )}{1 - \epsilon \|\hatbm{\Phi}_{uy}\|_2}\,,
\end{align*}
and 
\begin{align*}
% (1-\epsilon \|\hatbm{\Phi}_{uy}\|_2)^2
        &\quad \|\hatbm{\Phi}_{yy}(I-\mathbf{\Delta}\hatbm{\Phi}_{uy})^{-1}(\hatbf{G}+\mathbf{\Delta})\|_F^2\\ 
        &\leq \frac{1}{(1-\epsilon \|\hatbm{\Phi}_{uy}\|_2)^2}\left(\|\hatbm{\Phi}_{yu}\|_F^2 + 2\epsilon \|\hatbm{\Phi}_{yu}\|_F\|\hatbm{\Phi}_{yy}\|_F (2 + \|\hatbm{\Phi}_{uy}\|_2\|\hatbf{G}\|_2 ) + \left(\epsilon \|\hatbm{\Phi}_{yy}\|_F (2 + \|\hatbm{\Phi}_{uy}\|_2\|\hatbf{G}\|_2 )\right)^2
        \right)\\
        &\leq \frac{1}{(1-\epsilon \|\hatbm{\Phi}_{uy}\|_2)^2}\left(\|\hatbm{\Phi}_{yu}\|_F^2 + 2\epsilon \underbrace{\|\hatbm{\Phi}_{yy}\hatbf{G}\|_F}_{\leq \|\hatbm{\Phi}_{yy}\|_F \|\hatbf{G}\|_{2}}\|\hatbm{\Phi}_{yy}\|_F (2 + \|\hatbm{\Phi}_{uy}\|_2\|\hatbf{G}\|_2 ) + \left(\epsilon \|\hatbm{\Phi}_{yy}\|_F (2 + \|\hatbm{\Phi}_{uy}\|_2\|\hatbf{G}\|_2 )\right)^2
        \right)\\
        &\leq \frac{1}{(1-\epsilon \|\hatbm{\Phi}_{uy}\|_2)^2}\left(\|\hatbm{\Phi}_{yu}\|_F^2\|_2 + \|\hatbm{\Phi}_{yy}\|_F^2\left(2\epsilon  \|\hatbf{G}\|_2 (2 + \alpha\|\hatbf{G}\|_2 ) + \epsilon^2(2 + \alpha\|\hatbf{G}\|_2)^2\right)\right)\\
        &= \frac{1}{(1-\epsilon \|\hatbm{\Phi}_{uy}\|_2)^2}\left(\|\hatbm{\Phi}_{yu}\|_F^2\|_2 + \|\hatbm{\Phi}_{yy}\|_F^2 h(\epsilon,\alpha,\hatbf{G})\right)\,.
\end{align*}
Proceeding analogously, one can also prove that
\begin{align*}
    & \|\hatbm{\Phi}_{yy}(I-\mathbf{\Delta}\hatbm{\Phi}_{uy})^{-1}(\hatbm{y}_{free}+\bm{\delta}_0)\|_F \leq   \frac{\|\hatbm{\Phi}_{yy} \hatbf{y}_{free}\|_F + \epsilon \|\hatbm{\Phi}_{yy}\|_F (2 + \|\hatbm{\Phi}_{uy}\|_2\|\hatbf{y}_{free}\|_2 )}{1 - \epsilon \|\hatbm{\Phi}_{uy}\|_2}\,,\\
    & \|\hatbm{\Phi}_{uy}(I-\mathbf{\Delta}\hatbm{\Phi}_{uy})^{-1}(\hatbm{y}_{free}+\bm{\delta}_0)\|_F \leq   \frac{\|\hatbm{\Phi}_{uy} \hatbf{y}_{free}\|_F + \epsilon \|\hatbm{\Phi}_{uy}\|_F (2 + \|\hatbm{\Phi}_{uy}\|_2\|\hatbf{y}_{free}\|_2 )}{1 - \epsilon \|\hatbm{\Phi}_{uy}\|_2}\,,\\
    & \|\hatbm{\Phi}_{yy}(I-\mathbf{\Delta}\hatbm{\Phi}_{uy})^{-1}(\hatbm{y}_{free}+\bm{\delta}_0)\|_F^2 \leq \frac{1}{(1-\epsilon \|\hatbm{\Phi}_{uy}\|_2)^2}\left(\|\hatbm{\Phi}_{yy} \hatbf{y}_{free}\|_F^2 + \|\hatbm{\Phi}_{yy}\|_F^2 h(\epsilon,\alpha,\hatbf{y}_{free})\right) \,,\\
    & \|\hatbm{\Phi}_{uy}(I-\mathbf{\Delta}\hatbm{\Phi}_{uy})^{-1}(\hatbm{y}_{free}+\bm{\delta}_0)\|_F^2 \leq \frac{1}{(1-\epsilon \|\hatbm{\Phi}_{uy}\|_2)^2}\left(\|\hatbm{\Phi}_{uy} \hatbf{y}_{free}\|_F^2 + \|\hatbm{\Phi}_{uy}\|_F^2 h(\epsilon,\alpha,\hatbf{y}_{free})\right) \,.
\end{align*}
Therefore, combining the above inequalities we finally conclude that
\begin{align*}
    &J(\mathbf{G},\mathbf{K})\\
    % &\leq \frac{1}{1-\epsilon \|\hatbm{\Phi}_{uy}\|_2}\sqrt{\|\hatbm{\Phi}_{yy}\|_F^2 (1+ h(\epsilon,\alpha,\hatbf{G}) + 
    % h(\epsilon,\alpha,\hatbf{y}_{free})) + 
    % \|\hatbm{\Phi}_{uy}\|_F^2 h(1 + \epsilon,\alpha,\hatbf{y}_{free}) + \|\hatbm{\Phi}_{yu}\|_F^2 + \|\hatbm{\Phi}_{uu}\|_F^2 + + \|\hatbm{\Phi}_{yy}\hatbf{y}_{free}\|_F^2 + \|\hatbm{\Phi}_{uy}\hatbf{y}_{free}\|_F^2}\\
    &\leq \frac{1}{1-\epsilon \|\hatbm{\Phi}_{uy}\|_2}\sqrt{\norm{\begin{bmatrix}
        	  \hatbm{\Phi}_{yy} & \hatbm{\Phi}_{yu} & \hatbm{\Phi}_{yy}\hatbf{y}_{free}\\\hatbm{\Phi}_{uy} & \hatbm{\Phi}_{uu} & \hatbm{\Phi}_{uy}\hatbf{y}_{free}
        	 \end{bmatrix}}_F^2 +
    \|\hatbm{\Phi}_{yy}\|_F^2 (h(\epsilon,\alpha,\hatbf{G}) +
    h(\epsilon,\alpha,\hatbf{y}_{free})) + 
    \|\hatbm{\Phi}_{uy}\|_F^2 h(\epsilon,\alpha,\hatbf{y}_{free})\,.
   }
\end{align*}

\subsection{Proof of Lemma~\ref{le:feasible}}
First, it is easy to verify that $\tildebm{\Phi}$ satisfies the affine constraints in~\eqref{prob:quasi_convex}; indeed, $\tildebm{\Phi}$ is defined to be the closed-loop responses when we apply $\mathbf{K}^\star$ to the estimated plant $\hatbf{G}$. Then, since $\eta< \frac{1}{5}$, it is easy to verify that $\widetilde{\gamma} \leq \epsilon^{-1}$. It remains to  show that $\normtwo{\tildebm{\Phi}_{uy}} \leq \min(\widetilde{\gamma},\alpha)$: it holds

\begin{align*}
    \normtwo{\tildebm{\Phi}_{uy}} &= \normtwo{\starbm{\Phi}_{uy}(I+\bm{\Delta} \starbm{\Phi}_{uy})^{-1}}\\
    &\leq \frac{\normtwo{\starbm{\Phi}_{uy}}}{1-\epsilon \normtwo{\starbm{\Phi}_{uy}}} \leq  \sqrt{2}\frac{\normtwo{\starbm{\Phi}_{uy}}}{1-\epsilon \normtwo{\starbm{\Phi}_{uy}}}\\
    &= \sqrt{2} \frac{\eta}{\epsilon(1-\eta)} = \widetilde{\gamma}\leq \alpha\,.
\end{align*}

\subsection{Proof of Theorem~\ref{th:suboptimality}}
The key of the proof is to find a useful relationship between $J(\mathbf{G},\hatbf{K}^\star)$ and $J(\mathbf{G},\starbf{K})$, by exploiting the fact that we know a suboptimal solution to \eqref{prob:quasi_convex} by Lemma~\ref{le:feasible}.   Using the assumption $\eta <\frac{1}{5}$ so that $ \alpha \geq \frac{5\sqrt{2}}{4} \normtwo{\starbm{\Phi}_{uy}} \geq \frac{\sqrt{2} \norm{\bm{\Phi}^\star_{uy}}_2}{1-\eta}=\sqrt{2}\frac{\eta}{\epsilon(1-\eta)} = \widetilde{\gamma}$,  we have
\begin{align*}
    % J(\mathbf{G},\starhatbf{K}) &= \norm{\begin{bmatrix}
    %     	  \starhatbm{\Phi}_{yy} & \starhatbm{\Phi}_{yu} & \starhatbm{\Phi}_{yy}\mathbf{y}_{free}\\\starhatbm{\Phi}_{uy} & \starhatbm{\Phi}_{uu} & \starhatbm{\Phi}_{uy}\mathbf{y}_{free}
    %     	 \end{bmatrix}}_F\\
    J(\mathbf{G},\starhatbf{K})
        	 &\leq \frac{1}{1-\epsilon \gamma^\star}\norm{\begin{bmatrix}
        	  \sqrt{1+h(\epsilon,\alpha,\hatbf{G}) + h(\epsilon,\alpha,\hatbf{y}_{free})}\starhatbm{\Phi}_{yy} & \starhatbm{\Phi}_{yu} & \starhatbm{\Phi}_{yy}\hatbf{y}_{free}\\\sqrt{1+ h(\epsilon,\alpha,\hatbf{y}_{free})}\starhatbm{\Phi}_{uy} & \starhatbm{\Phi}_{uu} & \starhatbm{\Phi}_{uy}\hatbf{y}_{free}
        	 \end{bmatrix}}_F\\
        	 &\leq \frac{1}{1-\epsilon \widetilde{\gamma}}\norm{\begin{bmatrix}
        	  \sqrt{1+h(\epsilon,\alpha,\hatbf{G}) + h(\epsilon,\alpha,\hatbf{y}_{free})}\tildebm{\Phi}_{yy} & \tildebm{\Phi}_{yu} & \tildebm{\Phi}_{yy}\hatbf{y}_{free}\\\sqrt{1+ h(\epsilon,\alpha,\hatbf{y}_{free})}\tildebm{\Phi}_{uy} & \tildebm{\Phi}_{uu} & \tildebm{\Phi}_{uy}\hatbf{y}_{free}
        	 \end{bmatrix}}_F\,,
\end{align*}
where $\gamma^\star$ is optimal for \eqref{prob:quasi_convex}, and the second inequality holds because $(\gamma^\star,\starhatbm{\Phi})$ represents the optimal solution to \eqref{prob:quasi_convex} and $(\widetilde{\gamma},\tildebm{\Phi})$ is a suboptimal feasible solution of \eqref{prob:quasi_convex} by Lemma~\ref{le:feasible}.  Using the definition of $\tildebm{\Phi}$ from Lemma~\ref{le:feasible}, we now relate the term
\begin{equation*}
    \widetilde{C} = \norm{\begin{bmatrix}
        	  \sqrt{1+h(\epsilon,\alpha,\hatbf{G}) + h(\epsilon,\alpha,\hatbf{y}_{free})}\tildebm{\Phi}_{yy} & \tildebm{\Phi}_{yu} & \tildebm{\Phi}_{yy}\hatbf{y}_{free}\\\sqrt{1+ h(\epsilon,\alpha,\hatbf{y}_{free})}\tildebm{\Phi}_{uy} & \tildebm{\Phi}_{uu} & \tildebm{\Phi}_{uy}\hatbf{y}_{free}
        	 \end{bmatrix}}_F\,,
\end{equation*}
to the optimal cost of problem~\eqref{prob:IOP}. By defining
\begin{equation*}
    M = h(\epsilon,\alpha,\hatbf{G}) + h(\epsilon,\alpha,\hatbf{y}_{free})+ h(\epsilon,\norm{\starbm{\Phi}_{uy}}_2,\mathbf{G}) + h(\epsilon,\norm{\starbm{\Phi}_{uy}}_2,\mathbf{y}_{free})\,,
\end{equation*}
and
\begin{equation*}
    V = h(\epsilon,\alpha,\hatbf{y}_{free})+h(\epsilon,\norm{\starbm{\Phi}_{uy}}_2,\mathbf{y}_{free})\,,
\end{equation*}
we derive
\begin{align*}
    \widetilde{C} &= \sqrt{ \norm{\begin{bmatrix}
        	  \tildebm{\Phi}_{yy} & \tildebm{\Phi}_{yu} & \tildebm{\Phi}_{yy}\hatbf{y}_{free}\\\tildebm{\Phi}_{uy} & \tildebm{\Phi}_{uu} & \tildebm{\Phi}_{uy}\hatbf{y}_{free}
        	 \end{bmatrix}}_F^2 + \left(h(\epsilon,\alpha,\hatbf{G}) + h(\epsilon,\alpha,\hatbf{y}_{free})\right)\norm{\tildebm{\Phi}_{yy}}_F^2+ h(\epsilon,\alpha,\hatbf{y}_{free})\norm{\tildebm{\Phi}_{uy}}_F^2}\\
        	 &\leq  \frac{1}{1-\epsilon\norm{\starbm{\Phi}_{uy}}_2}\sqrt{J(\mathbf{G},\starbf{K})^2+ M \norm{\starbm{\Phi}_{yy}}_F^2 + V \norm{\starbm{\Phi}_{uy}}_F^2}\,,
        	% &\leq  \frac{1}{1-\epsilon\norm{\starbm{\Phi}_{uy}}_2}\norm{\begin{bmatrix}
        	  %\sqrt{1+M}\starbm{\Phi}_{yy} & \starbm{\Phi}_{yu} & \starbm{\Phi}_{yy}\hatbf{y}_{free}\\\sqrt{1+ V}\starbm{\Phi}_{uy} & \starbm{\Phi}_{uu} & \starbm{\Phi}_{uy}\hatbf{y}_{free}
        	% \end{bmatrix}}_F\,,
\end{align*}
where the bound 
\begin{align*}
   (1-\epsilon \norm{\starbm{\Phi}_{uy}}_2)^2 \norm{\begin{bmatrix}
        	  \tildebm{\Phi}_{yy} & \tildebm{\Phi}_{yu} & \tildebm{\Phi}_{yy}\hatbf{y}_{free}\\\tildebm{\Phi}_{uy} & \tildebm{\Phi}_{uu} & \tildebm{\Phi}_{uy}\hatbf{y}_{free}
        	 \end{bmatrix}}_F^2\leq &J(\mathbf{G},\mathbf{K}^\star)^2+( h(\epsilon,\norm{\starbm{\Phi}_{uy}}_2,\mathbf{G}) + h(\epsilon,\norm{\starbm{\Phi}_{uy}}_2,\mathbf{y}_{free}))\norm{\starbm{\Phi}_{yy}}_F^2+\\
        	 &+h(\epsilon,\norm{\starbm{\Phi}_{uy}}_2,\mathbf{y}_{free})\norm{\starbm{\Phi}_{uy}}_F^2\,,
\end{align*}
is derived in the same way as in Lemma~\ref{le:upperbound}, by using the expressions in Lemma~\ref{le:feasible}.

Thus, we have established the chain of inequalities
\begin{align*}
    J(\mathbf{G},\starhatbf{K}) \leq \frac{1}{1-\epsilon \widetilde{\gamma}} \widetilde{C} \leq \frac{1}{1-\epsilon \widetilde{\gamma}} \frac{1}{1-\epsilon\norm{\starbm{\Phi}_{uy}}_2}\sqrt{J(\mathbf{G},\starbf{K})^2+ M \norm{\starbm{\Phi}_{yy}}_F^2 + V \norm{\starbm{\Phi}_{uy}}_F^2}\,.
\end{align*}
%By connecting the above two inequalities 
Taking the squares, recalling that $\eta < \frac{1}{5}$, 
and using the fact that if $M,V>0$, then
\begin{equation*}
    Ma^2+Vb^2 \leq (M+V)(a^2+b^2)\,,
\end{equation*}
we derive
\begin{align*}
    \frac{J(\mathbf{G},\starhatbf{K})^2-J(\mathbf{G},\starbf{K})^2}{J(\mathbf{G},\starbf{K})^2}
     &\leq \left(\frac{1}{(1-\epsilon \normtwo{\starbm{\Phi}_{uy}})^2(1-\epsilon \widetilde{\gamma})^2}\right) \left(1+\frac{M\norm{\starbm{\Phi}_{yy}}_F^2+V\norm{\starbm{\Phi}_{uy}}_F^2}{J(\mathbf{G},\mathbf{K}^\star)^2}\right)-1\\
    %&\leq \left(\left(\frac{1}{1-(1+\sqrt{2})\epsilon_B\normtwo{\starbm{\Phi}_{uy}}}\right)^2-1\right)+\frac{h(\epsilon,\normtwo{\starbm{\Phi}_{uy}},\normtwo{\starbfcal{O}})+h(\epsilon,\alpha,\normtwo{\hatbfcal{O}})}{(1-(1+\sqrt{2})\epsilon_B\normtwo{\starbm{\Phi}_{uy}})^2}\\
    &\leq  \eta \left(\frac{2(1+\sqrt{2})-(1+\sqrt{2})^2 \eta}{(1-(1+\sqrt{2})\eta)^2}\right) +\frac{M\norm{\starbm{\Phi}_{yy}}_F^2+V\norm{\starbm{\Phi}_{uy}}_F^2}{J(\mathbf{G},\mathbf{K}^\star)^2(1-(1+\sqrt{2})\eta)^2}\\
    &\leq   \eta \left(\frac{2(1+\sqrt{2})-(1+\sqrt{2})^2 \eta}{(1-(1+\sqrt{2})\eta)^2}\right) +\frac{M+V}{(1-(1+\sqrt{2})\eta)^2}\\
    &\leq   20\eta +4(M+V)\,.%\\
    %&\leq \mathcal{O}(\epsilon)\,.
    %&\leq \eta \left(\frac{2(1+\sqrt{2})}{(1-(1+\sqrt{2})\eta)^2}\right) + 4\left(h(\epsilon,\normtwo{\starbm{\Phi}_{uy}},\normtwo{\starbfcal{O}})+h(\epsilon,\alpha,\normtwo{\hatbfcal{O}})\right)\\
    %&\leq 20\eta  + 4\left(h(\epsilon,\normtwo{\starbm{\Phi}_{uy}},\normtwo{\starbfcal{O}})+h(\epsilon,\frac{\sqrt{2}}{4},\normtwo{\hatbfcal{O}})\right)\,.
\end{align*}
Last, we prove that $20\eta + 4(M+V) = \mathcal{O}\left(\epsilon \norm{\starbm{\Phi}_{uy}}_2(\norm{\mathbf{G}}^2_2+\norm{\mathbf{y}_{free}}^2_2)\right)$.  By considering the expressions of $M$ and $V$, using $\alpha\leq 5\norm{\starbm{\Phi}_{uy}}_2$, $\eta<\frac{1}{5}$, $\norm{\hatbf{G}} \leq \norm{\mathbf{G}}+\epsilon$ and $\norm{\hatbf{y}_{free}} \leq \norm{\mathbf{y}_{free}}+\epsilon$, we deduce that:
\begin{align*}
    &M = h(\epsilon,\alpha\,\hatbf{G}) + h(\epsilon,\alpha,\hatbf{y}_{free})+ h(\epsilon,\norm{\starbm{\Phi}_{uy}}_2,\mathbf{G})+ h(\epsilon,\norm{\starbm{\Phi}_{uy}}_2,\mathbf{y}_{free})\\
    &\leq 2\Big[\epsilon^2(2 \hspace{-0.05cm}+\hspace{-0.05cm} 5\norm{\starbm{\Phi}_{uy}}_2\|\mathbf{G}\|_2 )^2+2\epsilon\norm{\mathbf{G}}_2(2\hspace{-0.05cm}+\hspace{-0.05cm}5\norm{\starbm{\Phi}_{uy}}_2\norm{\mathbf{G}}_2)+\epsilon^2(2 + 5\norm{\starbm{\Phi}_{uy}}_2\|\mathbf{y}_{free}\|_2 )^2+\\
    &~~~+2\epsilon\norm{\mathbf{y}_{free}}_2(2+5\norm{\starbm{\Phi}_{uy}}_2\norm{\mathbf{y}_{free}}_2)\Big] +\mathcal{O}(\epsilon^2 \norm{\starbm{\Phi}_{uy}}_2(\norm{\mathbf{G}}^2_2 + \norm{\mathbf{y}_{free}}^2_2))\\
    &= \mathcal{O}\left(\epsilon \norm{\starbm{\Phi}_{uy}}_2(\norm{\mathbf{G}}^2_2 + \norm{\mathbf{y}_{free}}^2_2)\right)\,,
\end{align*}
%\begin{equation*}
%    h(a, b,\mathbf{Y}) = a^2(2 + b\|\mathbf{Y}\|_2 %)^2+2a\norm{\mathbf{Y}}_2(2+b\norm{\mathbf{Y}}_2) = a(2 + %b\|\mathbf{Y}\|_2 )(a(2 + b\|\mathbf{Y}\|_2 ) + 2 %\norm{\mathbf{Y}}_2)
%\end{equation*}
and similarly $V= \mathcal{O}\left(\epsilon \norm{\starbm{\Phi}_{uy}}_2\norm{\mathbf{y}_{free}}^2_2\right)$. The result follows.

 %h(a, b,\mathbf{Y}) = a^2(2 + b\|\mathbf{Y}\|_2 )^2+2a\norm{\mathbf{Y}}_2(2+b\norm{\mathbf{Y}}_2)
}

\end{document}

%% file: interconnection.tex
\setlength{\unitlength}{0.008in}
\begin{picture}(243,170)(140,410)
\thicklines

%LOWER-LEFT Corner
\put(147,440){\vector(1, 0){ 93}}
\put(141,540){\vector( 0, -1){ 95}}
\put(88,440){\vector( 1, 0){ 48}}
\put(141,440){\circle{10}}
%\put(180,510){\line( 1, 0){ 48.5}}
\put(170,445){\makebox(0,0)[lb]{$\mathbf{y}$}}
\put(100,445){\makebox(0,0)[lb]{$\mathbf{v}$}}
\put(145,450){\makebox(0,0)[lb]{\tiny +}}
\put(127,445){\makebox(0,0)[lb]{\tiny +}}

%UPPER-RIGHT Corner

\put(385,540){\circle{10}}
\put(380,540){\vector( -1, 0){ 25}}
\put(385,439.5){\vector( 0,1){ 95.5}}
\put(438,540){\vector(-1, 0){ 48}}
\put(410,545){\makebox(0,0)[lb]{$\mathbf{w}$}}
\put(365,545){\makebox(0,0)[lb]{$\mathbf{u}$}}
\put(220,545){\makebox(0,0)[lb]{$\mathbf{x}$}}
\put(395,545){\makebox(0,0)[lb]{\tiny +}}
\put(390,528){\makebox(0,0)[lb]{\tiny +}}
\put(266,440){\line(1, 0){ 119.5}}

\put(240,427){\framebox(25,25){}} %K box
\put(245,435){\makebox(0,0)[lb]{$\mathbf{K}$}}

\put(330,527){\framebox(25,25){}} %B box
\put(335,535){\makebox(0,0)[lb]{$B$}}

\put(304,542){\makebox(0,0)[lb]{\tiny +}}
%\put(300,547.5){\makebox(0,0)[lb]{\tiny +}}
\put(300,527.5){\makebox(0,0)[lb]{\tiny +}}

%%INNER CICLE
\put(329,540){\vector(-1, 0){26}}
\put(298,540){\circle{10}}
%\put(298,580){\vector(0, -1){35}}
%\put(301,570){\makebox(0,0)[lb]{$\delta_{\mathbf{x}}$}}
\put(293,540){\vector(-1, 0){25}}

\put(243,527){\framebox(25,25){}} %z^-1 box
\put(248,537){\makebox(0,0)[lb]{\tiny  $z^{\text{-}1}$}}

\put(242,540){\vector(-1, 0){46}}
\put(219,540){\line(0, -1){46}}
\put(218.5,494){\vector(1, 0){24}}

\put(243,482){\framebox(25,25){}} %z^-1 box
\put(248,490){\makebox(0,0)[lb]{$A$}}

\put(268.5,494){\line(1, 0){29.7}}
\put(298.5,493.5){\vector(0, 1){42}}

%C Part
\put(171,527){\framebox(25,25){}} %z^-1 box
\put(176,535){\makebox(0,0)[lb]{$C$}}
\put(170.5,540){\line(-1, 0){30}}

%RED BOX?

{\color{red}
\put(160,475){\dashbox(203,90){}} %B box
\put(339,480){\makebox(0,0)[lb]{$\mathbf{G}$}}
}

\end{picture}